\documentclass[a4paper,reqno]{amsart}

    \usepackage[english]{babel} 

    \usepackage{hyperref} 
    \usepackage{graphicx} 
    \usepackage[monochrome]{color} 
    \usepackage{csquotes} 
    \usepackage{array} 

    \usepackage{amsmath,amssymb,amsfonts,amsthm} 
    \usepackage{mathtools,todonotes} 
    \usepackage{tikz-cd}  
    \usepackage[all,cmtip]{xy} 
    \usepackage[bbgreekl]{mathbbol} 

\theoremstyle{definition} 
    \newtheorem{definition}{Definition}

\theoremstyle{plain} 
    \newtheorem{theorem}[definition]{Theorem}
    \newtheorem{proposition}[definition]{Proposition}
    \newtheorem{lemma}[definition]{Lemma}
    \newtheorem{corollary}[definition]{Corollary}

\theoremstyle{remark} 
    \newtheorem{remark}[definition]{Remark}

\usepackage[bibstyle=alphabetic,citestyle=alphabetic,useprefix,giveninits=true, minbibnames=99,maxbibnames=99,sorting=nyt, sortcites=true,backend=biber]{biblatex}
\renewbibmacro{in:}{} 

\DeclareSourcemap{
  \maps[datatype=bibtex]{
    \map[overwrite]{ 
      \step[fieldsource=doi, final]
      \step[fieldset=url, null]
      \step[fieldset=eprint, null]
      }
     \map[overwrite]{ 
      \step[fieldsource=eprint, final]
      \step[fieldset=pages, null]
      \step[fieldset=eid, null]
      \step[fieldset=journal, null]
    }  
  }
}

\bibliography{bibliography}  
\newcommand{\Ker}[1]{\mathrm{Ker}{(#1)}}
\newcommand{\Fo}{{F_{\omega_0}}}
\newcommand{\Do}{{d_{\omega_0}}}
\DeclareMathOperator{\Ima}{Im}

\title{Gravitational Constraints on a Lightlike boundary}

\author{G. Canepa}
\address{Institut f\"ur Mathematik, Universit\"at Z\"urich, Winterthurerstrasse 190, 8057 Z\"urich, Switzerland}
\email{giovanni.canepa@math.uzh.ch}

\author{A. S. Cattaneo}
\address{Institut f\"ur Mathematik, Universit\"at Z\"urich, Winterthurerstrasse 190, 8057 Z\"urich, Switzerland}
\email{cattaneo@math.uzh.ch}

\author{M. Tecchiolli}
\address{}
\email{manuel.tecchiolli@gmail.com}

\thanks{This research was (partly) supported by the NCCR SwissMAP, funded by the Swiss National Science Foundation. G.C. and A.S.C. acknowledge partial support of SNF Grant No. 200020- 192080/1.}

\begin{document}

\begin{abstract}
We analyse the boundary  structure of General Relativity in the coframe formalism in the case of a lightlike boundary, i.e.,  when the restriction of the induced Lorentzian metric to the boundary is degenerate.
We describe the associated reduced phase space in terms of constraints on the symplectic space of boundary fields. We explicitly compute the Poisson brackets of the constraints and identify  the first- and second-class ones.  In particular, in the 3+1 dimensional case, we show that the reduced phase space has two local degrees of freedom, instead of the usual four in the non-degenerate case.

\end{abstract}

\maketitle
\tableofcontents

\section{Introduction}
The field-theoretical formulation of General Relativity (GR) is the assignment to a manifold M of an action
functional depending on a Lorentzian metric, whose Euler--Lagrange equations are Einsteins equations. If we now consider a manifold $M$ (of dimension $N$) with boundary $\partial M = \Sigma$ a natural question that can be raised is the structure of the induced data of field equations on the boundary $\Sigma$.
 This structure can be described through the \emph{reduced phase space} of the theory which encodes the data of the space of boundary fields and of the constraints of the theory. 
 
{ In this paper we study the reduced phase space of General Relativity (GR) in the coframe formulation in the case where the boundary has a light-like induced metric. The corresponding geometric structures for the space-like and time-like cases have already been studied by two of the authors in \cite{CCS2020}, based on the results outlined in \cite{CS2019}. The differences between the cases are given by the signature of the restriction of the metric to the boundary.
Indeed it turns out that there are major differences between the cases when the metric is \emph{space-like} or \emph{time-like} ---respectively with signature as a symmetric bilinear form $(N-1,0,0)$ or $(N-2,1,0)$ where the first index denotes positive eigenvalues, the second negative ones and the third zero ones--- and when the metric is \emph{light-like} ---with signature $(N-2,0,1)$ where the last entry refers to the transversal direction. Note that, since the metric in the bulk is Lorentzian, the metric on the boundary can only be non-degenerate or have a unique direction along which it is degenerate.

In this paper, following the same scheme of \cite{CS2019,CCS2020},} the boundary structure is recovered through a method that was firstly described by Kijowski and
Tulczijew (KT) in \cite{KT1979} opposed to the one proposed by Dirac \cite{Dirac1958}. This latter approach to the problem at hand has been developed in \cite{AlexandrovSpeziale15}. This article stems from the observations in \cite{CS2019, CCS2020} and describes the geometric structure of the boundary fields by adapting the result to the case of a degenerate boundary metric.    In $3+1$ dimensions, this results in a reduced phase space with two local degrees of freedom (in good agreement with the literature \cite{AlexandrovSpeziale15}) instead of four in the non-degenerate case\footnote{By number of local degrees of freedom we mean the rank { of the phase space} as a $C^\infty$-module (ignoring global degrees). In the space-like or time-like cases, one also usually speaks of the number of local physical degrees of freedom meaning by this half the rank of the reduced phase space (i.e., the rank of the configuration space).}.

The advantages of the KT alternative, in which the reduced phase space is described as a reduction { (i.e. as a quotient space)} of the space of free boundary fields, reside principally in the simplification of the procedure that leads to the definition of the constraints starting from the restriction of the Euler-Lagrange equations in the bulk.  Furthermore this construction avoids the introduction of the artificial classifications of the constraints as primary, secondary, etc. Another important virtue of this approach is its compatibility with the BV-BFV construction (\cite{CMR2012}), whose quantization procedure (\cite{CMR2}) can then be more easily  applied to the theory. The BV-BFV formalism provides a procedure to construct the reduced phase space too; however, it is not applicable in this case for $N\geq 4$ (\cite{CS2019})  since some of the regularity assumptions fail to be satisfied. It is worth noting that the present paper treats only this case, since the case $N=3$ has already been successfully analysed in \cite{CaSc2019} and does not display the issues of the higher dimensional case.

As mentioned above, in this paper we consider the coframe formulation. More precisely we use the Palatini-Cartan (PC) formalism (from \cite{Palatini1919, Cartan}) since its formulation through differential forms and connection is very convenient for the boundary (and corner \cite{C2021}) analysis. The choice of the formalism is not immaterial  due to the fact that classically equivalent theories on the bulk can behave differently in the presence of a boundary \cite[Section 4.3]{CS2019}. This is the case of gravity, where the space of solutions of the Euler-Lagrange equations (modulo symmetries) of the PC and the Einstein-Hilbert formulations are isomorphic, but their Hamiltonian formulations present striking, although classically irrelevant, differences, in particular in the structure of their BV-BFV formalism (\cite{CS2016b, CS2017}). 
{ The Ashtekar formalism} provides yet another alternative way through which this problem has been studied in the literature \cite{Ashtekar1986, ILV2006}; however, we will not explore this direction. Furthermore the same problem can be analyzed in greater generality such as for example the one proposed in \cite{FP1990} (where no compatibility with either the coframe or the internal metric is required) and the parent formulation proposed in \cite{BGparent1}, but we postpone the comparison with them to future works.

One of the gratest challenges of the constraint analysis of the PC theory comes from the structure of the symplectic form of the \emph{true} space of boundary { fields. It is a quotient space of the restriction of the bulk fields  to the boundary under an equivalence relation depending on the coframe. Since the use of equivalence classes is  usually quite annoying to handle, it is useful to fix a representative and describe the reduced phase space with it.  This has been done for a space-like or time-like boundaries in \cite{CCS2020} through the introduction of a suitable \emph{structural constraint}. However, such constraint has to be adapted in the light-like case, since it fixes the representative only provided that the induced metric on the boundary is degenerate.  In this paper we extend the solution proposed for the space- and time-like cases to a light-like boundary by considering a suitable adaptation. In particular, the key point is to modify the \emph{structural constraint}. The solution that we find is slightly more involved and gives rise to second class constraints, as opposed to the non-degenerate case where all constraints are first class.  The analysis} is carried out in full generality for every dimension $N\geq 4$. 

Furthermore we propose a linearized version of the theory, in Appendix \ref{sec:linearized_theory}, where we work around a reference solution of the Euler-Lagrange equation. In this case it can be shown that there is a natural isomorphism between the quotient space of the space of fields and another space where no equivalence classes are taken into account. This leads to a large simplification of the computations still retaining some of the key features of the real boundary theory, thus being also a nice toy model for the general case. In order to keep the results as simple and clean as possible this part has only been developed for $N=4$, but it can be extended without problems to higher dimensions.

The importance of this problem is witnessed by the number of previous works considering the structure of GR on null foliations, the first of which date back to Penrose and Sachs \cite{Sachs1962,Penrose1980}.
In particular the description of the Hamiltonian formulation of GR in the case of a null hypersurface has been studied for example in \cite{Torre1985,PS2017}  and in \cite{Reisenberger2013,Reisenberger2018} in the Einstein--Hilbert formalism . This formulation would allow the construction of exact (but not unique) solutions  starting from initial data on null hypersurfaces such as for example null horizons of black holes. Furthermore a Hamiltonian formulation of the theory is widely considered to be one of the best starting points for the quantization of the theory.

\noindent
\textbf{Acknowledgements.}
Part of this paper is the result of the master thesis of Manuel Tecchiolli at ETH-Zurich. We thank Michele Schiavina and Simone Speziale for all the interesting discussions and invaluable suggestions.
{ We also thank the anonymous referees for the comments and suggestions to improve the paper.}

\subsection{Structure of the paper} 
The last sections of this Introduction are devoted to recollecting the background material and reviewing the results of the paper. 

In Section \ref{sec:technical} we state most of the technical results needed throughout the paper. The proofs are collected in Appendix \ref{sec:appendix_long_proofs} for completeness, but can be skipped by the hasty reader.

The past results and the formal introduction to the problem motivating this work are collected in Section \ref{sec:boundary_known_results}. In particular we recall the main results of the non-degenerate case as stated in \cite{CCS2020}.

Finally, in Section \ref{sec:general_case} we consider the general case and illustrate in full detail the boundary structure of the degenerate case. The main results are collected in Theorem \ref{thm:Brackets_constraints}.

In Appendix \ref{sec:linearized_theory} we develop the corresponding linearized theory which is a simpler toy model of the general case. The structure of the \emph{linearized constraints} is in Theorem \ref{thm:Poissonbrackets_lin_d}.

\subsection{Palatini--Cartan formalism}

In this section we present the Palatini--Cartan formalism (see for example \cite{T2019, thi2007} and references therein for a review of the classical structure) and state the relevant (for our construction) results of \cite{CS2019}. For a more detailed description, we refer to \cite[Section 2]{CCS2020}.

We consider an $N$-dimensional oriented smooth manifold $M$ together with a Lorentzian structure
so that we can reduce the frame bundle to an $SO(N-1,1)$-principal bundle $P \rightarrow M$. 
We denote by $\mathcal{V}$ the associated vector bundle by the standard representation. Each fibre of $\mathcal{V}$ is isomorphic to an $N$-dimensional vector space $V$ with a Lorentzian inner product $\eta$ on it. The inner product allows the identification $\mathfrak{so}(N-1,1) \cong \bigwedge^2 {V}$. Furthermore we use the shortened notation
\begin{align}\label{e:NotationOmega}
    \Omega^{i,j}:= \Omega^i\left(M, \textstyle{\bigwedge^j} \mathcal{V}\right)
\end{align}
{ to indicate the spaces of $i$-forms on $M$ with values in the $j$th wedge product of $\mathcal{V}$.\footnote{In a language more common in the physics literature, using index notation, we can say that we can equip an element in $\Omega^{i,j}$ with $i$ contravariant indices (antisymmetrized in the cotangent space of $M$ and $j$ antisymmetrized indices in $\mathcal{V}$.}
Moreover we define the wedge product on these spaces as a map
\begin{align*}
    \wedge : \Omega^{i,j} \times &  \Omega^{k,l}  \rightarrow  \Omega^{i+k,j+l} & & \text{for} \; i+k \leq N, \; j+l \leq N \\
    (\alpha, & \beta)  \mapsto  \alpha \wedge \beta & & 
\end{align*}
by taking the wedge product on both the external ($T^*M$) and internal ($\mathcal{V}$) parts\footnote{Using index notation this map corresponds to taking antisymmetrization in both set of indices. Note also that the combinatorial factor arising in such operation is absorbed in the definition of wedge product and will not appear in formulas without indices.}. When no confusion can arise, we will omit the wedge symbol and consider it as understood (i.e., any expression of the form $\alpha\beta$ should be interpreted as $\alpha\wedge \beta$).
}

The dynamical fields of the theory are a $P$-connection $\omega$ and a coframe $e$ { (a.k.a  $N$-bein)}, i.e., an orientation preserving bundle isomorphism covering the identity 
\begin{align*}
    e \colon TM \stackrel{\sim}{\longrightarrow}\mathcal{V} .
\end{align*}
From the coframe it is possible to recover a metric as
\begin{equation}\label{e:g}
g_{\mu\nu}=\eta(e_\mu,e_\nu).
\end{equation}
The space of the $P$-connections, denoted with $\mathcal{A}(M)$, can be identified, via choosing a reference connection $\omega_0$, to $\Omega^{1,2}$ thanks to  $\mathfrak{so}(N-1,1) \cong \bigwedge^2 V$. We denote by $d_\omega$ and by $F_\omega \in \Omega^{2,2}$ respectively the covariant derivative $\Omega^{\bullet,\bullet}\to\Omega^{\bullet+1,\bullet}$ associated to a connection $\omega$ and its curvature.

The action functional of the theory is \footnote{ Note that the quantities appearing in this integral are elements of $\Omega^{N,N}$ which can be canonically identified with the space of densities on $M$, hence this integral is well defined. This same observation holds for every integral appearing in the paper. See \cite{CCS2020} for a detailed explanation.
}
\begin{align}\label{e:PCaction}
    S = \int_M\left[\frac1{(N-2)!}e^{N-2}F_\omega - \frac1{N!}\Lambda e^N\right]
\end{align}
where the notation $e^{k}$ denotes the $k$th wedge power of $e$ and $\Lambda$ is a constant (the cosmological constant). From the action we can deduce the Euler--Lagrange (EL) equations of the theory by taking its variations. The EL equation corresponding to the variation of $\omega$ is $d_\omega(e^{N-2})=0$, and using the Leibniz rule this equation can be rewritten as $e^{N-3}d_\omega e= 0$, which in turn, as we will see with  Lemma \ref{lem:We_bulk}, is equivalent to
\begin{equation}\label{e:tf}
d_\omega e = 0.
\end{equation}
The Euler--Lagrange equation corresponding to the variation of $e$ is 
\begin{equation}\label{e:ee}
\frac1{(N-3)!}e^{N-3}F_\omega-\frac1{(N-1)!}\Lambda e^{N-1}=0.
\end{equation}

Equation \eqref{e:tf} is the torsion-free condition and identifies the connection $\omega$ with the Levi-Civita connection of the metric \eqref{e:g}. With this
substitution, \eqref{e:ee} corresponds then to the Einstein equations.

\subsection{Overview}

We present here the problem and the results of the paper at a qualitative level (and for $N=4$) and refer to the subsequent sections for a more precise treatment. 

{ The main contribution of this article, as mentioned in the Introduction, is the description of the reduced phase space of General Relativity in the PC formalism on light-like boundaries as the critical locus of functions (or constraints) defined on a symplectic space of boundary fields induced from the bulk structure.}

The starting point of { this description} is the boundary symplectic structure induced by the bulk action following the construction described by \cite{KT1979}. { This construction starts from the variation of the classical action and extracts a one-form on the space of the restrictions\footnote{For differential forms we might as well speak of pullback with respect to the inclusion of the boundary in the bulk.} of the fields to the boundary. Subsequently it is possible to get a closed two-form by taking the de~Rham differential (on the space of fields) of the original one-form. If this two-form is degenerate, it is then possible to construct a symplectic form\footnote{i.e. a closed, non-degenerate two-form.} by taking a quotient (under the assumption that the quotient space is smooth).} The upshot of the construction {in the Palatini--Cartan case,} described first in \cite{CS2019} { and recalled in detail at the beginning of Section \ref{sec:boundary_known_results}}, is that the symplectic space of the boundary theory is a quotient space  {$F_{PC}^{\partial}= \widetilde{F}_{PC}/_\sim$} where the elements of {$\widetilde{F}_{PC}$} are the restrictions of the coframe $e$ and the connection $\omega$ to the boundary\footnote{\label{foot:pullback-restriction}We will use the same symbols for the fields on the bulk and the corresponding pullbacks (or restrictions) to the boundary.} and the equivalence relation is given by $\omega \sim \omega + v$, with $v$ satisfying $e \wedge v=0$. The resulting symplectic form is 
\begin{align*}
{\varpi_{PC}^{\partial} }= \int_{\Sigma} e \delta e \delta [\omega].
\end{align*} 
Now, in order to pass from the symplectic space of boundary fields, or \emph{geometric phase space}, to the \emph{reduced phase space}, we must identify the correct constraints of the theory.
The natural candidates for the constraints on the boundary are the restrictions of the Euler--Lagrange equations that contain no derivatives transversal to the boundary
\begin{align*}
    d_\omega e = 0 \quad \text{and} \quad e F_\omega-\frac1{6}\Lambda e^{3}=0.
\end{align*}
However, these functions are not invariant under the change of representative in the aforementioned quotient space. {Indeed, let us consider the first equation and consider two different $\omega \sim \omega'$, i.e., $\omega= \omega'+v$ with $ev=0$. The equation $d_\omega e = 0$ does not necessarily imply $d_{\omega'} e = 0$ since we get an additional term: $d_\omega e = d_{\omega'} e + [v,e]$ and  in general $[v,e]\neq 0$ for $v \in \Omega_{\partial}^{1,2}$ such that $ev=0$.}

In \cite{CCS2020} a convenient solution was found in the case of non-degenerate boundary metric whereas in \cite{CS2019} a general solution is outlined. The object of this paper is to find an analogous solution in the degenerate case, and therefore to generalize the result of \cite{CS2019,CCS2020} to all possible boundary metrics.

The construction of the non-degenerate case is described in detail in Section \ref{sec:boundary_known_results} and consists on imposing an equation fixing a convenient representative of the equivalence class $[\omega]$:
\begin{equation}\label{e:structural_constraint_overview}
    e_n d_{\omega} e \in \Ima ( e \wedge \cdot).
\end{equation}
{ Here $e_n \in \Omega_{\partial}^{0,1}$ is a  field linearly independent from the tangent components of $e$ restricted to the boundary\footnote{
More precisely, note that, by the nondegeneracy condition on $e$, at each $u\in\Sigma$ we have that $e(T_u\Sigma)$ is an oriented, three-dimensional subspace of $\mathcal{V}_u$. The field $e_n$ is chosen so that $e_n(u)$ is transversal to $e(T_u\Sigma)$ and compatible with the orientation. Equivalently, if we pick local coordinates $(x^1,x^2,x^3)$ around $u$ and expand $e=e_1 dx^1+e_2 dx^2+e_3 dx^3$, then we require that $(e_1(u),e_2(u),e_3(u),e_n(u))$ be a frame for $\mathcal{V}_u$ at each $u\in\Sigma$.}. The rationale behind this condition is to partially reobtain a condition on bulk fields that is not transferred to the boundary fields. Indeed, one of the EL equations ($ed_{\omega}e=0$, in the bulk equivalent to $d_{\omega}e=0$) can be written in a neighbourhood of the boundary as an evolution equation: 
$e_n d_{\omega}e + e \partial_n e + e[\omega_n,e]+ e d_{\omega}e_n=0$ where the index $n$ denotes a component transversal to the boundary. It is then easy to see that since the last terms are all in the image of $e \wedge \cdot$, also the first term must be in this space. We can then use this condition on the boundary to fix the representative of the class $[\omega]$ (see Section \ref{sec:technical} for the notation and Theorem \ref{thm:omegadecomposition} for the details).}
We call this condition the \emph{structural constraint}\footnote{Note that this additional condition on the boundary fields is not required for the description of the boundary structure. However, it is useful for fixing a representative of the equivalence class $[\omega]$.}.

{ Using the representative fixed by \eqref{e:structural_constraint_overview}, it is then possible to write a set of constraints generating the same critical locus of the original ones and which are invariant as follows:}

\begin{align*}
    L_c & = \int_{\Sigma} c e d_{\omega} e, \\
    P_{\xi} &= \int_{\Sigma}  \iota_{\xi} e e F_{\omega} + \iota_{\xi} (\omega-\omega_0) e d_{\omega} e, \\
    H_{\lambda} &= \int_{\Sigma} \lambda e_n \left( e F_{\omega} +\frac{1}{3!}\Lambda e^{3} \right)
\end{align*}
{ where $c$, $\xi$ and $\lambda$ are suitable Lagrange multipliers. A very important bit of information is given by the structure of their Poisson brackets which} is collected in Theorem \ref{thm:first-class-constraints} { and shows that these constraints are  first-class}.

This solution, and in particular the choice of the structural constraint, requires that the induced metric $g^\partial= e^{*}\eta$ be non-degenerate { and does not work in the degenerate case. The adaptation of such approach to the degenerate case is the object of this paper and in the following paragraphs we will give an overview on how to overcome the differences of this case.}

{
\begin{remark}\label{rmk:how_lightlike}
In this paper we address the problem assuming that in the boundary manifold there exists a light-like subset and we  assume to be working only in an open subset of the light-like one.  The general case of a boundary with points of different types (lightlike, spacelike and timelike) can be recovered as explained below in Remark \ref{rmk:general_case}.
\end{remark}
}

{ The main difference in the degenerate case is the impossibility of finding a representative of the equivalence class $[\omega]$ satisfying the structural constraint. The idea is to modify this equation by subtracting the problematic part and impose a weakened structural constraint as follows:}

\begin{align}\label{e:structural_modified_overview}
    e_n  d_{\omega} e - e_n  p_{\mathcal{T}}(d_{\omega} e) \in  \Ima ( e \wedge \cdot)
\end{align}
where $p_{\mathcal{T}}$ is the projection to an appropriately defined subspace (see \eqref{e:defTKS}; see also Section \ref{sec:technical} for the notation and Theorem \ref{thm:omegadecomposition_deg} for more details). This weakened structural constraint no longer fixes  the representative in the equivalence class uniquely {and hence it has to be supplemented with another set of equations, though of little importance for the construction}. Furthermore { this weakened constraint} does not guarantee the equivalence between the constraint { $L_c$ and $d_{\omega}e=0$. Indeed, an important feature that was a key point in the non-degenerate case was the fact that the equation $ed_{\omega}e=0$, after imposing the  structural constraint $e_n d_{\omega}e= \in \Ima ( e \wedge \cdot)$, defines the same zero locus as $ d_{\omega} e=0$.} As a consequence, in order to get the correct reduced phase space, { in the degenerate case} one has to add an additional constraint {accounting for the missing part in the weakened structural constraint: namely,}
\begin{align*}
R_{\tau} = \int_{\Sigma} \tau  d_{\omega} e
\end{align*}
with {$\tau$ belonging to an appropriate space $\mathcal{S}$}(see \eqref{e:defS} for the definition).  We will call this constraint the \emph{degeneracy constraint}\footnote{We thank M. Schiavina for the helpful discussion about the form of this constraint (and its name).}. This construction is made precise in the first part of Section \ref{sec:general_case} where we also analyse the structure of this new set of constraints (Theorem \ref{thm:Brackets_constraints} and Corollary \ref{cor:constraints-first-second_class}). 

{By computing the Poisson brackets of the constraints,} we show that all the constraints are first class { except the degeneracy constraint $R_{\tau}$ which is second class}. Finally we also compute the number of local physical degrees of freedom of the theory.  In dimension 3+1 we obtain that the reduced phase space has two local degrees of freedom.

{
\begin{remark}\label{rmk:general_case}
This construction can be extended to the general case of a boundary only part of which is allowed to be light-like. In this case the field $\tau \in \mathcal{S}$ defining the degeneracy constraint has support in the closure of the light-like points.
Furthermore, since the equations defining $\tau \in \mathcal{S}$ are algebraic, by continuity we also have that $\tau$ vanishes on the boundary (if present) of the closed light-like subset. 
\end{remark}
}

The linearized theory follows a similar pattern. It retains the most important properties of the general theory (e.g. the number of physical local degrees of freedom) and can be therefore thought of as an interesting toy model of the latter. The complete analysis of this case has been detailed in Appendix \ref{sec:linearized_theory}. Furthermore the linearized case is treated in the physical case $N=4$ only, hence providing a simple reference for the formulas and results in this case.

\begin{table}[ht]
\begin{tikzcd}[row sep=0.35cm] 
(\mathcal{F}, S )
\arrow[d,dashed, "\text{induce}"]\\
(\check{\mathcal{F}}, \check{\varpi}, \check{\mathsf{C}})
\arrow[d, "\text{Fix representative}"]\\
(\mathcal{F}^{\partial}, {\varpi}^{\partial}, \mathsf{C}^{\partial})
\arrow[d, "\text{quotient by constraints}"]\\
\text{Reduced Phase Space}
\end{tikzcd} 
\caption{Step by step construction of the Reduced Phase Space}
\label{t:steps}
\end{table}

{
We can recollect the steps in the Table \ref{t:steps}. The starting point is the bulk structure, given by the space of fields $\mathcal{F}$ and the action $S$. Then we induce a preboundary structure $(\check{\mathcal{F}}, \check{\varpi}, \check{\mathsf{C}})$ where $\check{\mathsf{C}}$ represents the restriction of the EL equations to the boundary. Subsequently we fix a representative in the equivalence class of $[\omega]$ and obtain the geometric phase space $(\mathcal{F}^{\partial}, {\varpi}^{\partial})$ where the constraints $\mathsf{C}^{\partial}$ are well defined. Finally the reduces phase space is obtained as the quotient of the geometric phase space by the constraints.
}

\begin{table}[ht] 
    \centering
\def\arraystretch{2}
\begin{tabular}{|>{\centering}p{0.25\textwidth}||>{\centering}p{0.25\textwidth} |>{\centering\arraybackslash}p{0.25\textwidth}|}
    \hline
     & Nondegenerate case & Light-like case\\
     \hline
     \hline 
    Geometric phase space & $(\mathcal{F}^{\partial}, \varpi^{\partial})$ & $(\mathcal{F}^{\partial}, \varpi^{\partial})$ \\
    \hline
    Structural constraint & \eqref{e:structural_constraint_overview} & \eqref{e:structural_modified_overview} \\
    \hline
    Constraints & $L_c, P_{\xi}, H_{\lambda}$ & $L_c, P_{\xi}, H_{\lambda}, R_{\tau}$ \\
    \hline
\end{tabular}
\vspace{0.3cm}
\caption{Differences between the nondegenerate case and the light-like one}
\label{t:comparison}
\end{table}

{
We conclude the overview with the Table \ref{t:comparison} showing the differences between the nondegenerate case and the light-like one.
}

\section{Technical results}\label{sec:technical}
{ In this section we} define the relevant quantities and maps,  establish the conventions and summarize the technical results needed in the paper. { One of the goal of this section is to prove some mathematical results in order to make the subsequent construction more fluid and easy to read. Full proofs and detailed computations will be postponed to Appendix \ref{sec:appendix_long_proofs}.}

We first recall and introduce some useful shorthand notation. { We will denote by $\Sigma= \partial M$ the $(N-1)$-dimensional  boundary  of the manifold $M$ of dimension $N$.} Furthermore we will use the  notation $\mathcal{V}_{\Sigma}$  for the restriction of $\mathcal{V}$ to $\Sigma$. { Extending the notation introduced in \eqref{e:NotationOmega}, using the same conventions,} we also define 
\begin{align*}
    \Omega^{i,j}:= \Omega^i\left(M, \textstyle{\bigwedge^j} \mathcal{V}\right) \qquad \Omega_{\partial}^{i,j}:= \Omega^i\left(\Sigma, \textstyle{\bigwedge^j} \mathcal{V}_{\Sigma}\right).
\end{align*}

We define the number of degrees of freedom of the space $\Omega^{i,j}$ (and $\Omega_{\partial}^{i,j}$) as its dimension  as a $C^\infty$-module. We will sometimes simply denote this by \emph{dimension}.

The coframe $e$ viewed as an isomorphism $e\colon TM \rightarrow \mathcal{V}$ defines, given a set of coordinates on $M$, a preferred basis on $\mathcal{V}$. If we denote by $\partial_i$ the vector field in $TM$ corresponding to the coordinate $x_i$, we get a basis on $\mathcal{V}$ by $e_i:= e (\partial_i)$. On the boundary, since $T\Sigma$ has one dimension less than $\mathcal{V}_{\Sigma}$, we can complement the linear independent set $e_i$ with another independent vector that we will call $e_n$.  { We call this basis the \emph{standard basis} (this basis depends on a given coordinate system on $M$ (or $\Sigma$))} and, unless otherwise stated, the components of the fields will always be taken with respect to this basis. 

On $\Omega^{i,j}$ and $\Omega_{\partial}^{i,j}$ we define the following maps:
\begin{align*}
    W_{k}^{ (i,j)}: \Omega^{i,j}  & \longrightarrow \Omega^{i+k,j+k} \\ 
    X  & \longmapsto   X \wedge \underbrace{e \wedge \dots \wedge e}_{k-times}, \\ 
    W_{k}^{ \partial, (i,j)}: \Omega_{\partial}^{i,j}  & \longrightarrow  \Omega_{\partial}^{i+k,j+k}\\ 
    X  & \longmapsto  X \wedge \underbrace{e \wedge \dots \wedge e}_{k-times}.
\end{align*}
{
Recall that the elements of the Lie algebra $\mathfrak{so}(N-1,1)$ can be identified with the elements of $\Omega^{(0,2)}$ (or $\Omega_{\partial}^{(0,2)}$, depending on where we consider such elements). Hence the Lie brackets define a map 
\begin{align*}
    [\cdot, \cdot] : \Omega^{(0,2)} \times \Omega^{(0,2)} & \rightarrow \Omega^{(0,2)} \\
    (x,y) &\mapsto [x,y],
\end{align*}
and a similar one on $\Omega_{\partial}^{(0,2)}$. Combining this action with the wedge product we can define the following generalisation, denoted with the same symbol
\begin{align*}
    [\cdot, \cdot] : \Omega^{(i,2)} \times \Omega^{(k,2)} & \rightarrow \Omega^{(i+k,2)} & & \text{for} \; i+k \leq N\\
    (x,y) &\mapsto [x,y],
\end{align*}
which in coordinates reads
\begin{align*}
    [x,y]_{\mu_1 \dots \mu_{i+k}}^{a_1 a_2} = \sum_{\sigma_{i+k}} \text{sign}(\sigma_{i+k}) x_{\mu_{\sigma(1)} \dots \mu_{\sigma(i)}}^{a_1 a_3} y_{\mu_{\sigma(i+1)} \dots \mu_{\sigma(i+k)}}^{a_2 a_4}\eta_{a_3 a_4}. 
\end{align*}
}
Furthermore, generalizing the action of the Lie algebra $\mathfrak{so}(N-1,1)$ on $\mathcal{V}$ (or $\mathcal{V}_{\Sigma}$) we can also introduce the following maps:
\begin{align}\label{e:generalised_Lie_algebra_action}
    \varrho^{(i,j)} : \Omega_{\partial}^{i,j}  & \longrightarrow \Omega_{\partial}^{i+1,j-1} \\
    X & \longmapsto [X, e] . \nonumber
\end{align}
In coordinates they are defined  as
{
\begin{align*}
    X & \mapsto \sum_{\sigma_{i+1}} \text{sign}(\sigma_{i+1}) X_{\mu_{\sigma(1)}\dots \mu_{\sigma(i)}}^{a_{1} \dots a_{j}} \eta_{a_{j} b} e^{b}_{\mu_{\sigma(i+1)}}.
\end{align*}
}
In the next part of this section we will state some technical results. We refer { to the appendix of \cite{CCS2020} for fully exhaustive proofs}. As in the aforementioned article we use by convention the total degree\footnote{Other sign conventions are possible, for example the one with separate degrees. Different conventions lead to isomorphic vector spaces but not isomorphic algebras.} to fix the commutation relations between quantities in $\Omega^{i,j}$ and $\Omega_{\partial}^{i,j}$.  {For example, given two elements\footnote{Later we will also consider elements with \emph{ghost number}. This means that we consider an additional $\mathbb{Z}$-grading and the total degree will be the sum of all the degrees.}
$\alpha \in \Omega^{i,j}$ and $\beta \in \Omega^{k,l}$ of total degree $i+j$ and $k+l$ respectively, we have the following commutation rule:
\begin{align*}
    \alpha\beta = (-1)^{(i+j)(k+l)}\beta \alpha.
\end{align*}
}
The properties of the maps $W_{k}^{ (i,j)}$ and $W_{k}^{ \partial, (i,j)}$ do not depend on the degeneracy of $g^{\partial}$. Hence we have the following results (\cite{CCS2020,CS2019}):

\begin{lemma} \label{lem:We_bulk}
    Let $N=\mathrm{dim}(M)\geq 4$. Then 
    \begin{enumerate}
        \item $ W_{N-3}^{ (2,1)}$ is bijective;  \label{lem:We21}
        \item $ \mathrm{dim}\mathrm{Ker}W_{N-3}^{(2,2)}\not=0$. \label{lem:We22}
    \end{enumerate}
\end{lemma}

\begin{lemma} \label{lem:We_boundary}
    The maps $W_{k}^{ \partial, (i,j)}$ have the following properties for $N \geq 4$:
    \begin{enumerate}
        \item $W_{N-3}^{\partial, (2,1)}$ and $W_{N-3}^{\partial, (1,2)}$ are surjective but not injective; \label{lem:Wep21,12}
        \item $W_{N-3}^{\partial, (1,1)}$ is injective; \label{lem:Wep11}
        \item $\dim \mathrm{Ker} W_{N-3}^{\partial, (1,2)} = \dim \mathrm{Ker} W_{N-3}^{\partial, (2,1)}$;\label{lem:kernel12-21}
        \item $W_{N-4}^{\partial, (2,1)}$ is injective. ($N \geq 5$)\label{lem:We5p21}
\end{enumerate}
\end{lemma}
The following lemma is an extension of the corresponding ones in \cite{CCS2020} and in \cite{CS2019}.
All the proofs of the following results can be found in  Appendix \ref{sec:appendix_long_proofs}.
\begin{lemma}\label{lem:varrho12_deg}
If $g^\partial$ is degenerate with {$\dim \Ker{g^\partial}=1$}, then
$\varrho^{(1,2)} |_{\mathrm{Ker} W_{N-3}^{\partial, (1,2)}}$ has a kernel of dimension $\frac{N(N-3)}{2}$.
\end{lemma}

{
\begin{remark}\label{rmk:ideas_behind_technical}
These three lemmas express in a mathematical way the possibility of inverting the coframe $e$ when appearing in a wedge product or in the generalised Lie algebra action $\varrho$ of \eqref{e:generalised_Lie_algebra_action}. In particular (exemplifying only in dimension $N=4$) they give the answer to the following question: given an expression of the form $e\wedge X$ or $[e,X]$ for some $X$, is it possible to invert these expressions and get back $X$? The answer is that it depends on the space where $X$ is defined, and in the case of $\varrho$ on the degeneracy of the boundary metric $g^\partial$. For example if we have $X \in \Omega^{2,1}$, using Lemma \ref{lem:We_bulk}, we see that it is possible to define an inverse $``W^{-1}_1"$ such that $X= W^{-1}_1(e \wedge X)$. On the contrary, for $X \in \Omega_{\partial}^{2,1}$ , using Lemma \ref{lem:We_boundary}, such inversion is no longer possible in a unique way, meaning that $e\wedge X$ does not contain all the information that $X$ contained (or, said in another way, not all the components of $X$ appear in $e \wedge X$).
Note also that these maps do not appear in the three-dimensional case. Hence their properties give hints on the differences between the topological three-dimensional theory and the physical four-dimensional one.
\end{remark}
}

\subsection{Results for the degeneracy constraint}

In order to define the space to which the Lagrange multiplier of the degeneracy constraint  belongs, it is useful to consider the following construction.

If a metric $g^{\partial}$ is degenerate, we can find a vector field $X$ on $\Sigma$ such that $\iota_{X}g^{\partial}=0$. Using a reference metric $g_0$, we can complete the vector field $X_0$ (with $\iota_{X_0}g_0^{\partial}=0$) to a basis $X_0, Y_0^i$ of $TM$. If we then choose a coframe $e$ \emph{near} the original one, the same $Y^i_0$s would also be a completion of $X$ to a basis of $TM$.

Let now $\beta \in \Omega_{\partial}^{1,0}$ a one form such that $\iota_X \beta =1$. We then define $\widehat{e}= \beta \iota_X e$ and fix $\beta$ by requiring that $\widetilde{e} \coloneqq e - \widehat{e}$ satisfies\footnote{The fact that the required condition is sufficient and well defined will be analyzed later in Lemma \ref{lem:deftau_rho}. }
\begin{align*}
    \iota_{Y_0^1} \dots \iota_{Y_0^{N-2}}(\widetilde{e} \wedge e^{N-4} \wedge v )= 0 
\end{align*}
for all  $v\in \Omega_{\partial}^{1,2}$  such that $e^{N-3} \wedge v = 0$.
 Using this notation we can define another set of maps
\begin{align*}
     \widetilde{\varrho}^{(i,j)} : \Omega_{\partial}^{i,j}  & \longrightarrow \Omega_{\partial}^{i+1,j-1} \\
     X & \longmapsto [X, \widetilde{e}]
\end{align*}
which in coordinate reads
{
\begin{align*}
    X & \mapsto \sum_{\sigma_{i+1}} \text{sign}(\sigma_{i+1}) X_{\mu_{\sigma(1)}\dots \mu_{\sigma(i)}}^{a_{1} \dots a_{j}} \eta_{a_{j} b} \widetilde{e}^{b}_{\mu_{\sigma(i+1)}}.
\end{align*}

}

Let $J$ be a complement\footnote{ \label{foot:Orthogonal_complement} For example it is possible to obtain an explicit expression for the complement in the following way. Choose an arbitrary Riemannian metric on $\Sigma$ and extend it to $\Omega^{2,1}$. Then it is possible to view $J$ as the orthogonal complement of $\Ima \varrho^{(1,2)} |_{\mathrm{Ker} W_{N-3}^{\partial, (1,2)}}$ in $\Omega_{\partial}^{2,1}$ with respect to this Riemannian metric. This approach will be used in Appendix \ref{sec:appendix_long_proofs} to prove the Lemmas and Proposition below with the diagonal Riemannian metric.} of the space $\Ima \varrho^{(1,2)} |_{\mathrm{Ker} W_{N-3}^{\partial, (1,2)}}$ in $\Omega_{\partial}^{2,1}$.
We now consider the following spaces:
\begin{subequations}\label{e:defTKS}
\begin{align}
    \mathcal{T}&= \mathrm{Ker}W_{N-3}^{\partial (2,1)} \cap J \subset \Omega_{\partial}^{2,1},\\
    \mathcal{K}&= \mathrm{Ker}W_{N-3}^{\partial (1,2)} \cap \mathrm{Ker} \varrho^{(1,2)} \subset \Omega_{\partial}^{1,2},\\
    \mathcal{S}&= \mathrm{Ker}W_{1}^{\partial (N-3,N-1)} \cap \mathrm{Ker} \widetilde{\varrho}^{(N-3,N-1)}  \subset \Omega^{N-3,N-1}_{\partial}. \label{e:defS}
\end{align}
\end{subequations}

{
\begin{remark}\label{rmk:meaningofspaces}
Note that all these three spaces are zero in the non-degenerate case. In particular the fact that $\mathcal{K}$ is not zero in the degenerate case accounts for the existence of components of $\omega$ that do not appear either in the expression $ed_{\omega}e$ or $e_n d_{\omega}e$ but do appear in $d_{\omega}e$ (for $N=4$). Hence $\mathcal{K}$ represents the failure of the structural constraint to fix uniquely a representative in the equivalence class $[\omega]$. The space $\mathcal{T}$ is strictly related to $\mathcal{K}$ since it contains elements of $\mathrm{Ker}W_{N-3}^{\partial (2,1)}$ that cannot be generated by elements in $\mathrm{Ker}W_{N-3}^{\partial (1,2)}$ through $\varrho^{(1,2)}$. As a matter of fact, using coordinates, one can see that the components of $\Omega_{\partial}^{2,1}$ corresponding to $\mathcal{T}$ in the non-degenerate case are generated through $\varrho^{(1,2)}$ by the elements corresponding to $\mathcal{K}$ in $\Omega_{\partial}^{1,2}$. Finally. $\mathcal{S}$ plays the role of the dual of $\mathcal{T}$ as specified in Lemma \ref{lem:relationSandT}.
\end{remark}

}

We also denote by $p_{\mathcal{T}}: \Omega_{\partial}^{2,1} \rightarrow \mathcal{T}$, by $p_{\mathcal{K}}: \Omega_{\partial}^{1,2} \rightarrow \mathcal{K}$ and by $p_{\mathcal{S}}: \Omega_{\partial}^{N-3,N-1} \rightarrow \mathcal{S}$ some corresponding projections to them\footnote{\label{foot:Orthogonal_complement2} In order to define these projections we may proceed as in footnote \ref{foot:Orthogonal_complement} and define an orthogonal complement of these spaces and subsequently use the corresponding orthogonal projections.}. The spaces $\mathcal{T}$ and $\mathcal{K}$ are not empty because of the results of Lemmas \hyperref[lem:Wep21,12]{\ref*{lem:We_boundary}.(\ref*{lem:Wep21,12})} and \ref{lem:varrho12_deg}, while $\mathcal{S}$ is characterized by the following Proposition in which we also summarize the involved components, since they will be crucial in the computation of the Poisson brackets of the constraints.

\begin{proposition}\label{prop:components_of_tau}
    The dimension of $\mathcal{S}$ is
    \begin{align*}
        \dim \mathcal{S} = \frac{N(N-3)}{2}.
    \end{align*}
    {
 Let $p\in \Sigma$ and $U$ a neighbourhood of $p$ in which normal coordinates centered in $p$ are well defined. Then using such coordinates and the standard basis of $\mathcal{V}_{\Sigma}$,
the non-zero components of an element $\tau \in \mathcal{S}$ are 
}
   
\begin{align*}
    Y_{\mu}& :=\tau^{NN-1 \mu_1 \dots \mu_{N-3}}_{\mu_1 \dots \mu_{N-3}} \text{ where } \mu \neq \mu_1 \dots \mu_{N-3}, \\
    X_{\mu_1}^{\mu_2} & := \tau^{N N-1 \mu_3 \dots \mu_{N-2} \mu_{2}}_{\mu_3 \dots \mu_{N-2} \mu_1},
\end{align*}
satisfying
\begin{align*}
    \sum_{\mu=1}^{N-2} Y_{\mu} =0 \text{ and } 
    X_{\mu_1}^{\mu_2} = f(\widetilde{g}^{\partial}, X_{\mu_2}^{\mu_1}, Y_{\mu})
\end{align*}
for $\mu_1 < \mu_2$  and some linear function $f$ with $\widetilde{g}^{\partial} \coloneqq \eta (\widetilde{e},\widetilde{e}) $. 
\end{proposition}

The proof of this Proposition is postponed to Appendix \ref{sec:appendix_long_proofs}.
\begin{remark}
In order to compute the structure of the Poisson brackets between the constraints we will need to know the equations defining $\mathcal{S}$ not only point-wise but also in a small neighbourhood, since we will need to take derivatives.  Despite being in principle computable for every dimension, we do not need the explicit expression of $f$. It is also worth noting that in the base point $p$ of the normal coordinates the last set of equations reduces to 
\begin{align*}
    X_{\mu_1}^{\mu_2} =- X_{\mu_2}^{\mu_1}.
\end{align*}
\end{remark}

While the space $\mathcal{K}$ and $\mathcal{T}$ arise naturally while considering the symplectic reduction of the boundary two form, the importance of the space $\mathcal{S}$ resides in the following Proposition that shows that  $\mathcal{S}$ plays the role of a \emph{dual space} of $\mathcal{T}$.

\begin{lemma}\label{lem:relationSandT}
    Let $\alpha \in \Omega^{2,1}_\partial$. Then 
    \begin{align*}
        \int_{\Sigma} \tau \alpha =0 \:\: \forall \tau \in \mathcal{S} \Longrightarrow p_{\mathcal{T}}(\alpha)=0.
    \end{align*}
\end{lemma}

We conclude this section with a result that will be necessary in the computation of the Hamiltonian vector fields of the constraints and in their Poisson brackets.

\begin{lemma}\label{lem:[tau,e]inImW}
\begin{align*}
  \Ima \varrho^{(N-1,N-3)}|_{\mathcal{S}} \subset \Ima W_{N-3}^{\partial, (1,1)}  .
\end{align*}
\end{lemma}

\begin{corollary} \label{cor:components_of_W-1[tau,e]}
    The free components of $W_{N-3}^{-1}([\tau,e])$ are 
    \begin{align*}
        [W_{N-3}^{-1}([\tau,e])]_{\mu_1}^{\mu_2} & \propto X_{\mu_1}^{\mu_2} \\
        [W_{N-3}^{-1}([\tau,e])]_{\mu}^{\mu} & \propto Y_{\mu} 
    \end{align*}
    such that $\sum_{\mu=1}^{N-2} [W_{N-3}^{-1}([\tau,e])]_{\mu}^{\mu}=0 $ and $[W_{N-3}^{-1}([\tau,e])]_{\mu_1}^{\mu_2}=-[W_{N-3}^{-1}([\tau,e])]_{\mu_2}^{\mu_1}$.
\end{corollary}
{ The proofs of these lemmas and of the corollary are collected in Appendix \ref{sec:appendix_long_proofs}.}

\section{Boundary structure and known results}\label{sec:boundary_known_results}
{
In this section we give an overview about the symplectic boundary structure of Palatini--Cartan gravity induced from the bulk using the construction introduced by Kijowski and Tulczijew \cite{KT1979}. In other words we give a description of the \emph{geometric phase space}, i.e., the natural space of fields associated to the boundary before imposing the constraints, and describe
the symplectic reduction that produces the reduced phase space. Referring to the table \ref{t:steps} in the overview, we give information about the first step $(\mathcal{F},S) \rightarrow (\check{\mathcal{F},\check{\varpi}},\check{\mathsf{C}})$ and about the geometric phase space. This part is common to both the non-degenerate (space-like or time-like) case and degenerate cases (light-like).
}

We dedicate this section to the common framework of the two cases and to the non-degenerate one by recalling the most important steps and results. This will be particularly useful, since the analysis of the degenerate case will start from these results trying to solve the various issues arising from the different structural constraints that we will choose. In particular the crucial difference will come from the different outcome of Lemma \ref{lem:varrho12_deg} in the degenerate and non-degenerate cases.

The investigation of the Hamiltonian formulation follows, as explained before, the construction introduced by Kijowski and Tulczijew \cite{KT1979}. The starting point is the description of what we call \emph{geometric phase space} ${F}^{\partial}_{PC}$. This step is fully detailed in \cite{CS2019}. We consider the restriction of the fields $e$ and $\omega$ to the boundary $\Sigma$ and reinterpret them { respectively} as an injective bundle map $T\Sigma\to \mathcal{V}_{\Sigma}$ (that we will call \emph{boundary coframe}) and an orthogonal connection associated to $\mathcal{V}_{\Sigma}$.  We call $\widetilde{F}_{PC}$ the space of these fields, i.e. the space of the restriction of the bulk fields to the boundary. The key point of the construction is  to define a one-form on the space $\widetilde{F}_{PC}$ as the boundary term arising from the variation of the action  through the formula 
\begin{align*}
    \delta S = \mathcal{EL} + \pi^* \check{\alpha}
\end{align*}
where $\mathcal{EL}$ are the parts defining the Euler-Lagrange equation and $\pi$ is the restriction to the boundary.

In our case we get 
\begin{align*}
    \check{\alpha} = {
\frac{1}{(N-2)!}}\int_{\Sigma}  e^{N-2}\delta\omega.
\end{align*}
From this one-form it is possible to construct a closed two-form by applying the de Rham differential $\delta$ (of the space of fields):
\begin{align*}
    \check{\varpi}= \delta \check{\alpha}  = {
\frac{1}{(N-3)!}}\int_{\Sigma}  e^{N-3}\delta e\delta\omega.
\end{align*}
This two-form is a candidate to be a symplectic form on the space of boundary fields; however, it is degenerate, since the function 
$W_{N-3}^{\partial, (1,1)}$ has a non-zero kernel (Lemma \ref{lem:We_boundary}): the kernel is parametrized by the vector fields {
$X= v \frac{\delta}{\delta \omega} \in \mathfrak{X}(\widetilde{F}_{PC})$} with $v$ such that
\begin{equation}\label{e:Xomega}
e^{N-3} v=0.
\end{equation} 
In order to get a symplectic form, we can perform a symplectic reduction by quotienting along the kernel. The \emph{geometric phase space} of boundary fields, determined by the reduction 
\begin{equation}\label{e:presymplecticreduction}
\pi_{PC}\colon \widetilde{F}_{PC} \longrightarrow {F}^{\partial}_{PC},
\end{equation}
is then parametrized by the field $e$ and by the equivalence classes of $\omega$ under the relation $\omega \sim \omega + v$ with $v$ satisfying \eqref{e:Xomega}. We denote by $\mathcal{A}^{red}(\Sigma)$ the space of such equivalence classes. Then the symplectic form on  ${F}^{\partial}_{PC}$ is given by
\begin{align}\label{e:classical-boundary-symplform}
\varpi^{\partial}_{PC} = \int_{\Sigma} e^{N-3} \delta e \delta [\omega]
\end{align}
{where we dropped the unimportant prefactor 
$\frac{1}{(N-3)!}$.}

The symplectic space $({F}^{\partial}_{PC},\varpi^{\partial}_{PC})$ is the space on which we can define the constraints and subsequently perform a reduction over them to get the reduced phase space. The constraints are now to be recovered from the restriction of the Euler--Lagrange equation on the bulk to the boundary. In particular we have to consider those equations not containing derivatives in the transversal direction, i.e. the evolution equations.

However, some obstruction might occur. We performed a reduction to get the symplectic form \eqref{e:classical-boundary-symplform}, yet the restriction of the functions whose zero-locus defines the Euler--Lagrange equations might not be basic with respect to it { i.e. it might not be possible to write such restrictions in terms of the variables of the reduced symplectic space ${F}^{\partial}_{PC}$}.  This is exactly what happens in our case: a simple check shows that the candidates to be the constraints coming from \eqref{e:ee} are not invariant under the transformation $\omega \mapsto \omega +v$. The way out proposed in \cite{CCS2020} for the non-degenerate case is to fix a convenient representative of the equivalence class $[\omega]$ and work out the details with it. In the next section we will recap the strategy and present the most important steps. This will turn to be useful also in the degenerate case.

\subsection{Non-degenerate boundary metric} \label{sec:recap_non-deg_case}
We recall here the steps to get the reduced phase space in the non-degenerate case as developed in \cite{CCS2020}. We refer to this work for the proofs and details that are omitted here. 

As already mentioned, we define  $e_n$ as a section of $\mathcal{V}_{\Sigma}$ that is a completion of the basis $e_1,e_2,\dots , e_{N-1}$. Then we have the following two results:
\begin{lemma}\label{lem:Omega2,1_d4}
    Let now $g^{\partial}$ be non-degenerate and let $\alpha \in \Omega^{2,1}_\partial$. Then $\alpha=0$ if and only if
    \begin{align}\label{e:ConditionforOmega21_d4}
        \begin{cases}
            e^{N-3}\alpha =0 \\
            e_n e^{N-4}\alpha \in \Ima W_{N-3}^{\partial, (1,1)}
        \end{cases}.
    \end{align}
\end{lemma}

\begin{lemma}\label{lem:Omega2,2_d4}
    Let $\beta \in \Omega^{N-2,N-2}_\partial$. If $g^\partial$ is nondegenerate, there exists a unique $v \in \mathrm{Ker} W_{N-3}^{\partial, (1,2)}$ and a unique $\gamma \in \Omega_{\partial}^{1,1}$ such that 
    \begin{align*}
        \beta = e^{N-3} \gamma + e_n e^{N-4} [v, e].
    \end{align*}
\end{lemma}

The key idea is to use these results to fix a representative for the equivalence class $[\omega] \in \mathcal{A}^{red}(\Sigma)$ appearing in the symplectic form \eqref{e:classical-boundary-symplform}. Applying Lemma \ref{lem:Omega2,1_d4} to $\alpha= d_\omega e $ we get that the constraint (coming from the bulk) $d_\omega e = 0$ can be divided into the \emph{invariant constraint} $e^{N-3} d_\omega e = 0$ and the constraint
\begin{equation}\label{e:omegareprfix}
    e_n e^{N-4} d_{\omega} e \in \Ima W_{N-3}^{\partial,(1,1)},
\end{equation}
called \emph{structural constraint}. Then the following results proves that \eqref{e:omegareprfix} exactly fixes a representative of the aforementioned equivalence class without imposing further constraints.

\begin{theorem}[\cite{CCS2020}]\label{thm:omegadecomposition}
    Suppose that  $g^{\partial}$, the metric induced on the boundary, is nondegenerate. Given any $\widetilde{\omega} \in \Omega_{\partial}^{1,2}$, there is a unique decomposition 
    \begin{equation} \label{e:omegadecomp}
        \widetilde{\omega}= \omega +v
    \end{equation}
    with $\omega$ and $v$ satisfying 
    \begin{equation}\label{e:omegareprfix2}
        e^{N-3}v=0 \quad \text{ and } \quad  e_n e^{N-4} d_{\omega} e \in \Ima W_{N-3}^{\partial,(1,1)}.
    \end{equation}
\end{theorem}

\begin{corollary}
The field $\omega$ in the decomposition \eqref{e:omegadecomp} depends only on the equivalence class $[\omega] \in \mathcal{A}^{red}(\Sigma)$.
\end{corollary}

Having fixed the representative of the equivalence class of the connection, one considers the restriction of the Euler--Lagrange equations to the boundary to get the corresponding constraints. The wise choice of the structural constraint \eqref{e:omegareprfix} allows to construct the set of constraints on the boundary. Defining $c \in\Omega^{0,2}_\partial[1]$, $\xi \in\mathfrak{X}[1](\Sigma)$ and $\lambda\in \Omega^{0,0}_\partial[1]$ as (odd)\footnote{Such quantities are also sometimes referred to as Grassmann variables.} Lagrange multipliers, they read
\begin{subequations}\label{e:constraints}
\begin{equation}\label{e:constraintL}
L_c = \int_{\Sigma} c e^{N-3} d_{\omega} e ,
\end{equation}
\begin{equation}\label{e:constraintP}
P_{\xi}= \int_{\Sigma}  \iota_{\xi} e e^{N-3} F_{\omega} + \iota_{\xi} (\omega-\omega_0) e^{N-3} d_{\omega} e,
\end{equation}
\begin{equation}
H_{\lambda} = \int_{\Sigma} \lambda e_n \left( e^{N-3} F_{\omega} +\frac{1}{(N-1)!}\Lambda e^{N-1} \right),
\end{equation}
\end{subequations}
where $\omega_0$ is a reference connection\footnote{\label{foot:invariance_omega_0} The critical locus of these constraints does not depend on $\omega_0$, since it appears in \eqref{e:constraintP} in combination with an expression already present in \eqref{e:constraintL}.}.

\begin{remark}
We use here odd Lagrange multipliers $c$, $\xi$ and $\lambda$, following \cite{CCS2020}. The notation $[1]$ next to the symbol of the space to which these quantities belong denotes the shift to odd quantities. This convention does not modify the structure of the constraints and simplifies the computations and the notation. The version with even Lagrange multipliers can be easily derived from the present one. For example, let us consider $\{L_c,L_c\}$. This bracket denotes an antisymmetric quantity in which the odd variables are space holders. This means that going back to unshifted (i.e., even) variables, say, $\alpha,\beta$, a formula like
\begin{align*}
    \{L_c,L_c\}=-\frac12 L_{[c,c]}
\end{align*}
simply means
\begin{align*}
     \{L_\alpha,L_\beta\} =- L_{[\alpha,\beta]}. 
\end{align*}
\end{remark}

The following theorem describes the structure of the constraints:

\begin{theorem}[\cite{CCS2020}] \label{thm:first-class-constraints}
     Let $g^\partial$ be nondegenerate on $\Sigma$. Then, the functions $L_c$, $P_{\xi}$, $H_{\lambda}$ are well defined on ${F}^{\partial}_{PC}$ and define a coisotropic submanifold  with respect to the symplectic structure $\varpi^{\partial}_{PC}$. In particular they satisfy the following relations
    \begin{subequations}\label{brackets-of-constraints}
        \begin{eqnarray}
            \{L_c, L_c\} = - \frac{1}{2} L_{[c,c]} & 
            \{P_{\xi}, P_{\xi}\}  =  \frac{1}{2}P_{[\xi, \xi]}- \frac{1}{2}L_{\iota_{\xi}\iota_{\xi}F_{\omega_0}} \\
            \{L_c, P_{\xi}\}  =  L_{\mathcal{L}_{\xi}^{\omega_0}c} & 
            \{L_c,  H_{\lambda}\}  = - P_{X^{(a)}} + L_{X^{(a)}(\omega - \omega_0)_a} - H_{X^{(n)}} \\
            \{H_{\lambda},H_{\lambda}\}  =0 & 
            \{P_{\xi},H_{\lambda}\}  =  P_{Y^{(a)}} -L_{ Y^{(a)} (\omega - \omega_0)_a} +H_{ Y^{(n)}} 
        \end{eqnarray}
    \end{subequations}
    where 
    \begin{align*}
        \mathcal{L}_{\xi}^{\omega} A = \iota_{\xi} d_{\omega} A -  d_{\omega} \iota_{\xi} A \qquad A \in \Omega^{i,j}_{\partial}
    \end{align*} 
    and
    $X= [c, \lambda e_n ]$, $Y = \mathcal{L}_{\xi}^{\omega_0} (\lambda e_n)$ and $Z^{(a)}$, $Z^{(n)}$ are the components of $Z\in\{X,Y\}$ with respect to the frame $(e_a, e_n)$.\footnote{ It is useful to stress here the differences in the notation between the first constraint and the Lie derivative. The first is denoted with an italic $L$, while the second with a calligraphic $\mathcal{L}$.}
\end{theorem}

\section{Degenerate boundary structure} \label{sec:general_case}

In section \ref{sec:boundary_known_results} we presented the construction of the boundary structure in the non-degenerate case.  Let now $g^{\partial}$ be degenerate, i.e. admitting a vector field $X$ such that $\iota_X g^{\partial}=0$.
{\subsection{Fixing a representative} In this section we describe a possible way for fixing the freedom of the choice of the connection $\omega \in [\omega]$ , adapting the non-degenerate case presented in \cite{CCS2020} and summarized in section \ref{sec:recap_non-deg_case}.
 The main difference is that in the degenerate case, because of the different outcome of Lemma \ref{lem:varrho12_deg}, it is no longer possible to find an $\omega \in [\omega]$ such that $e_n e^{N-4}d_{\omega} e \in \Ima W_{N-3}^{\partial,(1,1)}$.} Indeed, in contrast to the non-degenerate case, the map 
\begin{align*}
v \in \Ker {W_{N-3}^{\partial,(1,2)} }\mapsto e_n e^{N-4} [v,e] \in \Omega_{\partial}^{N-2,N-2} 
\end{align*}
is not injective on $W_{N-3}^{\partial,(1,2)}$ (Lemma \ref{lem:varrho12_deg}). The workaround is to separately consider the components of $d_{\omega} e $ in $\mathcal{T}$ and the components of $\omega$ in $\mathcal{K}$ (where $\mathcal{T}$ and $\mathcal{K}$ are introduced in \eqref{e:defTKS}). Indeed in the following theorem we consider a weaker version of the structural constraint \eqref{e:omegareprfix} that generalizes it for a  degenerate metric. This theorem is the generalization of Theorem \ref{thm:omegadecomposition}. 

\begin{theorem}\label{thm:omegadecomposition_deg}
    Let $g^{\partial}$ be degenerate. Given any $\widetilde{\omega} \in \Omega_{\partial}^{1,2}$, there is a unique decomposition 
    \begin{equation} \label{e:omegadecomp_mod}
        \widetilde{\omega}= \omega +v
    \end{equation}
    with $\omega$ and $v$ satisfying
    \begin{subequations}\label{e:omegareprfix2_mod}
        \begin{align}
            & e^{N-3}v=0, \\
            & e_n e^{N-4} d_{\omega} e - e_n e^{N-4} p_{\mathcal{T}}(d_{\omega} e) \in \Ima W_{N-3}^{\partial,(1,1)}, \label{e:structural_constraint_mod}\\
          & p_{\mathcal{K}} v = 0.
        \end{align}
    \end{subequations}
\end{theorem}

The proof is based on the following two lemmas generalizing respectively Lemmas \ref{lem:Omega2,1_d4} and \ref{lem:Omega2,2_d4}.

\begin{lemma}\label{lem:Omega2,1_d4_mod}
Let $g^{\partial}$ be degenerate and let $\alpha \in \Omega^{2,1}_\partial$. Then $\alpha=0$ if and only if 
\begin{align}\label{e:ConditionforOmega21_d4_deg}
\begin{cases}
e^{N-3}\alpha =0 \\
e_n e^{N-4}\alpha  - e_n e^{N-4}p_{\mathcal{T}}\alpha \in \Ima W_{N-3}^{\partial, (1,1)}\\
p_{\mathcal{T}}\alpha = 0
\end{cases}.
\end{align}
\end{lemma}

\begin{proof}
    Trivial generalization of Lemma \ref{lem:Omega2,1_d4}.
\end{proof}

\begin{lemma}\label{lem:Omega2,2_d4_mod}
Let $\beta \in \Omega^{N-2,N-2}_\partial$. If $g^\partial$ is degenerate, there exist a unique $v \in \mathrm{Ker} W_{N-3}^{\partial, (1,2)}$, a unique $\gamma \in \Omega_{\partial}^{1,1}$ and a unique $\theta \in \mathcal{T}$ such that 
\begin{align*}
\beta = e^{N-3} \gamma + e_n e^{N-4} [v, e] + e_n e^{N-4} \theta.
\end{align*}
\end{lemma}

\begin{proof}
By definition of $\mathcal{T}$ it is clear that for each element $\alpha \in \mathrm{Ker}W_{N-3}^{\partial (2,1)}$ it is possible to find $\theta \in \mathcal{T}$ and $v \in \mathrm{Ker} W_{N-3}^{\partial, (1,2)}$ such that $\alpha= [v,e] + \theta.$ From the proof of Lemma \ref{lem:Omega2,2_d4} we also know that each element $\beta \in \Omega^{N-2,N-2}_\partial$ can be written as 
$\beta = e^{N-3} \gamma + e_n e^{N-4} \alpha$ for some $\alpha \in \mathrm{Ker}W_{N-3}^{\partial (2,1)}$. Combining these two results we get the claim.
\end{proof}

\begin{proof}[Proof of Theorem \ref{thm:omegadecomposition_deg}]
    Let $\widetilde{\omega} \in \Omega_\partial^{1,2}$. From Lemma \ref{lem:Omega2,2_d4_mod} we deduce that there exist $\sigma \in \Omega_\partial^{1,1}$, $v \in \text{Ker} W_1^{\partial,(1,2)}$ and $\theta \in \mathcal{T}$ such that 
\begin{align*}
e_n e^{N-4} d_{\widetilde{\omega}} e = e^{N-3} \sigma + e_n e^{N-4}[v,e] + e_n e^{N-4} \theta.
\end{align*}
We define $\omega := \widetilde{\omega} - v $. Then $\omega$ and $v$ satisfy \eqref{e:omegadecomp_mod} and \eqref{e:omegareprfix2_mod}.
\end{proof}

In contrast with the non-degenerate case, this theorem does not fix completely the freedom of $\omega \in [\omega]$. Hence we require the following additional equation: 
\begin{align}\label{e:second_structural_constraint}
p_{\mathcal{K}} \omega =0.
\end{align}
Hence \eqref{e:structural_constraint_mod} and \eqref{e:second_structural_constraint} fix uniquely the representative in the equivalence class\footnote{ Starting from the definition of $\mathcal{K}$ in \eqref{e:defTKS}, it is a straightforward check that this last equation fixes the components of $\omega \in \Ker{ W_{N-3}^{\partial,(1,2)}}$ not included in \eqref{e:structural_constraint_mod}. Indeed, the elements of $\omega \in \Ker \rho $ are the ones that no longer appear in the structural constraints in the degenerate case opposed to the non-degenerate one.}. 

{
\subsection{Independence from the choices}\label{sec:Independence}
In this section we explore the independence of the analysis from the choices that we have made in the construction. We prove it through the following general theorem.
\begin{theorem}
    Let $(P, \varpi)$ be a presymplectic manifold with kernel distribution $K$, smooth leaf space $(\underline{P}, \underline{\varpi})$ and canonical projection $\pi: P \rightarrow \underline{P}$. Let $Q$ be a submanifold of $P$ such that
    \begin{align*}
        \rho:= \pi |_{Q} : Q \rightarrow \underline{P}
    \end{align*}
    is a diffeomorphism. Then $(Q, \varpi|_{Q})$ is a symplectic manifold and $\rho$ is a symplectomorphism.
\end{theorem}
\begin{proof}
    For every $x\in P$ we have that the exact sequence
    \begin{align*}
        0 \rightarrow K_x \rightarrow T_x P \overset{d_x\pi}{\rightarrow} T_{\pi(x)} \underline{P} \rightarrow 0.
    \end{align*}
    For $x\in Q$ we have the splitting $d_{\rho(x)}: T_{\pi(x)}\underline{P}\rightarrow T_x P$ with image $T_x Q$ which gives $T_xM = T_x Q \oplus K_x$. Let now $v \in (T_xQ)^{\perp}$, then $\varpi_x(v,w)=0$ $\forall w \in T_xQ$. Furthermore $\varpi_x(v,w)= \varpi_x (v, w + \widetilde{w})$ for all $\widetilde{w} \in K_x$. From the previous result we get that 
    $\varpi_x (v, \widehat{w})=0$ for all $\widehat{w} \in T_x P$. This implies that $v \in (T_x P)^{\perp}=K_x$. Therefore $(T_xQ)^{\perp} \subseteq K_x$ and 
    \begin{align*}
        (T_xQ)^{\perp} \cap T_xQ \subseteq K_x \cap T_xQ = \emptyset. 
    \end{align*}
    Hence $(Q, \varpi|_{Q})$ is symplectic.
    
    From the definition of leaf space we have that 
    \begin{align*}
        \varpi_x(v,w) = \underline{\varpi}_{\pi(x)}([v],[w]) \qquad \forall x \in P \; \forall v,w \in T_x P
    \end{align*}
    Restricted to $Q$ this becomes
        \begin{align*}
        \varpi_x(v,w) = \underline{\varpi}_{\rho(x)}([v],[w]) \qquad \forall x \in Q \; \forall v,w \in T_x Q.
    \end{align*}
    Since $\rho$ is a diffeomorphism and $(Q, \varpi|_{Q})$ is a symplectic manifold, this last equation proves that $\rho$ is a symplectomorphim.
\end{proof}
\begin{corollary}
If $Q$ and $Q'$ are submanifolds of $P$ such that $\pi|_{Q}$ and $\pi|_{Q'}$ are diffeomorphisms with $\underline{P}$, then $(Q, \varpi|_{Q})$ and $(Q', \varpi|_{Q'})$ are canonically simplectomorphic.
\end{corollary}
\begin{remark}
In our case $P$ is the space of restrictions to the boundary $\widetilde{F}_{PC}$ with presymplectic form $\check{\varpi}$, $Q$ is the subspace of $\widetilde{F}_{PC}$ where $\omega$ satisfies the constraints \eqref{e:structural_constraint_mod} and \eqref{e:second_structural_constraint}, while $\underline{P}$ is the geometric phase space $F_{PC}^{\partial}$ with simplectic form $\varpi_{PC}^{\partial}$ defined in \eqref{e:classical-boundary-symplform}. The map $\pi$ is given by $\pi_{PC}$ defined in \eqref{e:presymplecticreduction} and $\rho$ is its restriction to $Q$. The inverse of $\rho$ is given by the map $(e,[\omega]) \mapsto (e, \omega')$ where $ \omega'$ is the unique representative of  the class $[\omega]$ satisfying \eqref{e:structural_constraint_mod} and \eqref{e:second_structural_constraint}.
\end{remark}

The existence of a canonical symplectomorphism between the constructions corresponding to different possible choices of the representative in the equivalence class of $[\omega]$ guarantees the independence of the construction on such choices. In particular the choice of the projection that leads to \eqref{e:second_structural_constraint} is immaterial in the construction since we do not use this constraints anywhere else.
}
{\subsection{Constraints of the theory}
Let} us now turn to the constraints of the theory. In the degenerate case we can still adopt the approach of the non-degenerate one adapting it to encompass the differences between Lemma \ref{lem:Omega2,1_d4} and Lemma \ref{lem:Omega2,1_d4_mod}.
The main difference is that now the constraint $L_c$ together with the new structural constraint \eqref{e:structural_constraint_mod} is no longer equivalent to $d_{\omega}e=0$ (one set of the Euler-Lagrange equations in the bulk) since we are missing the third equation in \eqref{e:ConditionforOmega21_d4_deg}. Indeed we have to add an additional constraint that, thanks to Lemma \ref{lem:relationSandT}, we can express as 
\begin{equation}\label{e:Constraint_R}
R_{\tau} = \int_{\Sigma} \tau  d_{\omega} e
\end{equation}
through an odd Lagrange multiplier $\tau \in \mathcal{S}[1]$\footnote{As before the notation $[1]$ denotes that $\tau$ is an odd quantity.}. Furthermore, to simplify the computation of the brackets between the constraints, it is useful to modify the constraint $H_{\lambda}$ by adding to it a term proportional to $R_{\tau}:$
\begin{equation}
    H_{\lambda} = \int_{\Sigma} \lambda e_n \left(\frac{1}{(N-3)}e^{N-3}F_\omega -e^{N-4}(\omega-\omega_0)p_{\mathcal{T}}(d_{\omega}e)+ \frac{1}{(N-1)!}\Lambda e^{N-1}\right).
\end{equation}
Note that we can as well express the second term in this constraint as
\begin{align*}
    \lambda p_{\mathcal{S}} ( e_n e^{N-4}(\omega-\omega_0))d_{\omega}e
\end{align*}
to make it explicitly in the form of \eqref{e:Constraint_R}.

\begin{remark} 
The additional part in $H_{\lambda}$ proportional to $R_{\tau}$ has been added only to ease the computation of the Hamiltonian vector field of the constraint $H_{\lambda}$ itself. Such a linear combination does not affect the constrained set and the structure of the constraints, i.e. the distinction between first and second class constraints (see Proposition \ref{prop:numberofsecondclassconstraints} and Remark \ref{rem:linearcomb_constraints_fsc} in Appendix \ref{sec:first-second_class_def}). Similar considerations hold also for the part of the constraint $P_{\xi}$ proportional to $L_c$, as already mentioned in \cite[Remark 4.24]{CS2019} and \cite[Remark 21]{CCS2020}.
\end{remark}

Before analysing the structure of these constraints and their Poisson brackets we need some additional results concerning the elements in $\mathcal{S}$ whose variations are constrained and are thus depending on $e$.

\begin{lemma}\label{lem:expression_deltatau_constrained}
The variation of an element $\tau \in \mathcal{S}$ is constrained by the following equations:
\begin{align*}
    p_{\widetilde{\rho}}'\delta \tau &= \widetilde{\rho}^{-1}\left(\frac{\delta \widetilde{\rho}}{\delta e } (\tau) \delta e\right), \\
    p_{W}'\delta \tau &= W_1^{-1}(\tau \delta e)
\end{align*}
where the inverses\footnote{ Note that, in order to avoid cumbersome notation, we will from now on avoid to write all the indices of the inverse functions of $W_{\bullet}^{\partial,(\bullet,\bullet)}$ and of $\widetilde{\rho}^{(\bullet,\bullet)}$.} are defined on their images and $p_{\widetilde{\rho}}'$ and $p_{W}'$ are respectively the projections to a complement of the kernel of $\widetilde{\rho}$ and $W_1^{\partial, (N-3,N-1)}$.
\end{lemma}
\begin{remark}
Different choices of projections lead to different terms in the kernel of the two maps. Nonetheless these additional terms are in $\mathcal{S}$ where the variation is free. Hence they will not play any role in the computations.
\end{remark}
\begin{proof}
From \eqref{e:defS} we know the elements $\tau \in \mathcal{S}$ must satisfy the following equations:
\begin{align*}
    \tau \wedge e =0; \quad  \widetilde{\rho}(\tau)=0.
\end{align*}
Hence varying each equation we obtain some constraints for the variation $\delta\tau$:
\begin{align*}
    \delta \tau \wedge e - \tau \wedge \delta e =0 ; \quad  \widetilde{\rho}(\delta \tau) + \frac{\delta \widetilde{\rho}}{\delta e } (\tau) \delta e =0.
\end{align*}
We can invert these equations using the inverses of $W_1^{\partial, (N-3, N-1)}$ and $\widetilde{\rho}$ on their images. Denoting with $p'_W$ and $p'_{\widetilde{\rho}}$  the projections to some complements of the kernel of $W_1^{\partial, (N-3, N-1)}$ and $\widetilde{\rho}$  in $\Omega_{\partial}^{N-3, N-1}$ respectively, we obtain
\begin{align*}
    p'_W \delta \tau = W_1^{-1} (\tau \wedge \delta e ) ; \quad p'_{\widetilde{\rho}} \delta \tau= \widetilde{\rho}^{-1}\left(\frac{\delta \widetilde{\rho}}{\delta e } (\tau) \delta e\right).
\end{align*}
These relations fix the constrained part of the variation of $\tau \in \mathcal{S}$ in terms of the variation of $e$.  
\end{proof}

\begin{lemma} \label{lem:computations_for_LR_and_PR}
The following identities hold:
\begin{align*}
    \widetilde{\rho}^{-1}\left(\frac{\delta \widetilde{\rho}}{\delta e } (\tau) [c,e]\right)=p'_{\widetilde{\rho}}[c, \tau], \qquad  \widetilde{\rho}^{-1}\left(\frac{\delta \widetilde{\rho}}{\delta e } (\tau) \mathcal{L}_{\xi}^{\omega_0} e\right)p'_{\widetilde{\rho}}\mathcal{L}_{\xi}^{\omega_0}\tau.
\end{align*}
\end{lemma}
\begin{proof}
    We start by making more explicit the expression $\widetilde{\rho}^{-1}\left(\frac{\delta \widetilde{\rho}}{\delta e } (\tau) \delta e\right)$.
    By definition, if $\tau \in \mathcal{S}$, then $[\tau, \widetilde{e}]=0$. Hence
    \begin{align*}
        0 = \delta [\tau, \widetilde{e}] = [\delta \tau, \widetilde{e}] + [\tau, \delta \widetilde{e}].
    \end{align*}
    We now compute $\delta \widetilde{e}$ in terms of $\delta e$:
    \begin{align*}
        \delta \widetilde{e}= \delta e - \delta \widehat{e}= \delta e - \delta (\beta \iota_X e)= \delta e - \delta \beta \iota_X e + \beta \iota_{\delta X} e -  \beta \iota_X \delta e.
    \end{align*}
    We have then to compute the variation $\delta X$ and $\delta \beta$. We start from the first: from the defining equation $\iota_X g^{\partial}=0$ we get
    \begin{align*}
        \iota_{\delta X} g^{\partial} - \iota_X \delta g^{\partial}=0
    \end{align*}
    and hence, inverting $g^{\partial}$ on its image, we get
    $\delta X= {g^{\partial}}^{-1}(\iota_X \delta g^{\partial})$. Since $g^{\partial}$ can be written in terms of $e$ and $\eta$ as $g^{\partial}= \eta(e,e)$, we can write this part of $\delta X$ in terms of $\delta e$.
    The remaining part of $\delta X$ not fixed by this equation is such that $\iota_{\delta X} g^{\partial}=0$, hence 
    \begin{align*}
        \delta X=  2{g^{\partial}}^{-1}(\iota_X \eta ( \delta e , e )) + \lambda X
    \end{align*}
    for some function $\lambda$.

    Let us now pass to $\delta \beta$. Its value is completely determined by the equations $\iota_{X}\delta \beta - \iota_{\delta X} \beta=0$ and     
    \begin{align*}
        \iota_{Y_0^1}\dots \iota_{Y_0^{N-2}}&\left(\iota_{\delta X} (\beta e^{N-3}) v - \iota_X(\delta \beta  e^{N-3}) v \right)\\ &+ \iota_{Y_0^1}\dots \iota_{Y_0^{N-2}}\left( (N-3)\iota_X(\beta  \delta e  e^{N-4}) v + \iota_X (\beta  e^{N-3}) \delta  v\right)= 0.
    \end{align*}
    This last equation must hold for every $v$ and $\delta v$ that satisfy respectively $e^{N-3} \wedge v = 0$ and $(N-3) \delta e e^{N-4} v + e^{N-3}\delta v=0$.

We can now plug  the values $\delta e= [c,e]$ and $\delta e= \mathcal{L}_{\xi}^{\omega_0} e$ in the first formula of Lemma \ref{lem:expression_deltatau_constrained} using the above results.
In the first case we get
\begin{align*}
    \delta X=  2{g^{\partial}}^{-1}(\iota_X [[c,e] , e ])+ \lambda X= 2{g^{\partial}}^{-1}(\iota_X [c,[e  , e ]])+ \lambda X= \lambda X
\end{align*}
and $ \delta \beta= \lambda \beta$. Consequently
\begin{align*}
    \widetilde{\rho}^{-1}\left([\tau, [c,e] -  \beta \iota_X [c,e]] \right)
    &=\widetilde{\rho}^{-1}\left([\tau, [c,e] -   [c, \beta\iota_X e]] \right)\\
    &= \widetilde{\rho}^{-1}\left([\tau, [c,\widetilde{e}] \right)\\
    &= \widetilde{\rho}^{-1}\left([[\tau, c],\widetilde{e}] +[c, [\tau, \widetilde{e}]]\right)\\
    &= p'_{\widetilde{\rho}} [\tau, c].
\end{align*}
In the second case we have
\begin{align*}
    \delta X=  2{g^{\partial}}^{-1}(\iota_X [\mathcal{L}_{\xi}^{\omega_0} e , e ])+ \lambda X= {g^{\partial}}^{-1}(\iota_X \mathcal{L}_{\xi}^{\omega_0} g^{\partial})+ \lambda X.
\end{align*}
and $\delta \beta= \mathcal{L}_{\xi}^{\omega_0} \beta+ \lambda \beta$.
In coordinates we obtain the following expressions
\begin{align*}
  \delta X^{\mu}&= X^{\rho} \partial_{\rho} \xi^{\mu} + \xi^{\rho} \partial_{\rho} X^{\mu} + \lambda X^{\mu}\\
    \iota_X \mathcal{L}_{\xi}^{\omega_0} e &= X^{\rho} \xi^{\mu}{ d_{\omega_0}}_{\mu} e_{\rho} - X^{\rho} e_{\mu} d_{\rho} \xi^{\mu}.
\end{align*}
Hence 
\begin{align*}
    \iota_X \mathcal{L}_{\xi}^{\omega_0} e + \iota_{\delta X} e= \iota_{\xi} d_{\omega} (\iota_X e )+ \lambda \iota_X e ,
\end{align*}
and collecting all these formulas we get
\begin{align*}
    \widetilde{\rho}^{-1}\left(\frac{\delta \widetilde{\rho}}{\delta e } (\tau) \mathcal{L}_{\xi}^{\omega_0} e\right)
    &= \widetilde{\rho}^{-1}\left([\tau,\mathcal{L}_{\xi}^{\omega_0} e - \mathcal{L}_{\xi}^{\omega_0} (\beta \iota_X e )]\right)\\
    &= \widetilde{\rho}^{-1}\left([\tau,\mathcal{L}_{\xi}^{\omega_0} \widetilde{e}] \right)\\
    &= \widetilde{\rho}^{-1}\left(\mathcal{L}_{\xi}^{\omega_0}[\tau, \widetilde{e}] - [\mathcal{L}_{\xi}^{\omega_0}\tau, \widetilde{e}]\right)\\
    & =  p'_{\widetilde{\rho}}\mathcal{L}_{\xi}^{\omega_0}\tau.
\end{align*}
\end{proof}

The addition of the constraint $R_{\tau}$ to compensate the different structure of the lightlike case has important consequences on the structure of the set of constraints. 

\begin{theorem}\label{thm:Brackets_constraints}
    Let $g^\partial$ be degenerate on $\Sigma$. Then the structure of the Poisson brackets of the constraints $L_c$, $P_{\xi}$, $H_{\lambda}$ and $R_{\tau}$ is given by the following expressions:
    \begin{align*}
    &\begin{aligned}
        	&\{L_c, L_c\}  = - \frac{1}{2} L_{[c,c]} & \qquad\qquad
       	& \{P_{\xi}, P_{\xi}\}  =  \frac{1}{2}P_{[\xi, \xi]}- \frac{1}{2}L_{\iota_{\xi}\iota_{\xi}F_{\omega_0}} \\
       	& \{L_c, P_{\xi}\}  =  L_{\mathcal{L}_{\xi}^{\omega_0}c} &
       	& \{H_{\lambda},H_{\lambda}\}  \approx F_{\tau' \tau'} \\
       	&  \{L_c, R_{\tau}\} =  -R_{p_{\mathcal{S}}[c, \tau]} & 
       	& \{P_{\xi},R_{\tau}\}  = R_{p_{\mathcal{S}}\mathcal{L}_{\xi}^{\omega_0}\tau}.\\
       	&\{ R_{\tau}, H_{\lambda} \}   \approx F_{\tau \tau'} + G_{\lambda \tau}  &
 		& \{R_{\tau},R_{\tau}\}  \approx F_{\tau \tau}
        \end{aligned}\\
	&\{L_c,  H_{\lambda}\}
        = - P_{X^{(a)}} + L_{X^{(a)}(\omega - \omega_0)_a} - H_{X^{(n)}} +  R_{p_{\mathcal{S}}(X^{(a)}e_a e^{N-4} (\omega- \omega_0)-\lambda e_n d_{\omega_0} c)} \\        
        &\{P_{\xi},H_{\lambda}\} =  P_{Y^{(a)}} -L_{ Y^{(a)} (\omega - \omega_0)_a} +H_{ Y^{(n)}}-R_{p_{\mathcal{S}}(Y^{(a)} e_a e^{N-4} (\omega- \omega_0) -\lambda e_n \iota_ {\xi}F_{\omega_0})} 
    \end{align*}
    where $\tau'=  p_{\mathcal{S}}(\lambda e_n e^{N-4} (\omega-\omega_0))$, $X= [c, \lambda e_n ]$, $Y = \mathcal{L}_{\xi}^{\omega_0} (\lambda e_n)$ and $Z^{(a)}$, $Z^{(n)}$ are the components of $Z\in\{X,Y\}$ with respect to the frame $(e_a, e_n)$. Furthermore $F_{\tau \tau}$, $F_{\tau \tau'}$, $F_{\tau' \tau'}$ and $G_{\lambda \tau}$ are functions of $e$, $\omega$, $\tau $ (or $\tau'$) and $\lambda$ defined in the proof that are not proportional to any other constraint.
\end{theorem}

\begin{remark}
In Theorem \ref{thm:Brackets_constraints} we use the symbol $\approx$ to denote the fact that the result can be obtained only working on shell, i.e. imposing the constraints. Here we want to stress that the brackets are not proportional to the constraints, while in the other cases (the ones with the $=$ sign) we get an exact result. Equivalently we could have written e.g. $\{L_c, L_c\}  \approx 0$.
\end{remark}

\begin{proof}
We first compute the variation of the constraints in order to find their Hamiltonian vector fields. Using the results of \cite{CCS2020} for $L_c$ and $P_{\xi}$, we have:
\begin{align*}
 \delta L_c= \int_{\Sigma}  -\frac{1}{N-2} c [\delta \omega, e^{N-2}]  + \frac{1}{N-2} c d_{\omega}\delta (e^{N-2}) = \int_{\Sigma} [c, e]  e^{N-3} \delta \omega  +  d_{\omega} c  e^{N-3}\delta e;
\end{align*}
\begin{align*}
 \delta P_{\xi} &= \int_{\Sigma}  \iota_{\xi} (e^{N-3} \delta e) F_{\omega}  -\frac{1}{N-2}\iota_{\xi} (e^{N-2}) d_{\omega}\delta \omega + \iota_{\xi} \delta \omega e^{N-3} d_{\omega} e  \\
 & \qquad-\frac{1}{N-2} \iota_{\xi} (\omega-\omega_0) [\delta \omega, e^{N-2}]
  + \frac{1}{N-2} \iota_{\xi} (\omega-\omega_0) d_{\omega}\delta (e^{N-2}) \\
 & = \int_{\Sigma}  - e^{N-3} \delta e (\mathcal{L}_{\xi}^{\omega_0} (\omega-\omega_0) + \iota_ {\xi}F_{\omega_0})  -  (\mathcal{L}_{\xi}^{\omega_0} e ) e^{N-3} \delta \omega ;
\end{align*}
\begin{align*}
    \delta R_{\tau} = &\int_{\Sigma} \delta_e \tau d_{\omega}e - \tau     [\delta \omega,e] + \tau d_{\omega} \delta e \\
    = & \int_{\Sigma} \delta e g(\tau,\omega ,e) + [\tau , e]  \delta \omega  + d_{\omega} \tau  \delta e
\end{align*}
where $g(\tau,\omega ,e)$ is a formal expression that encodes the dependence of $\delta \tau$ on $\delta e$ i.e. such that 
\begin{align*}
    \delta e g(\tau,\omega ,e) = p'_{\widetilde{\rho}} \widetilde{\rho}^{-1}\left(\frac{\delta \widetilde{\rho}}{\delta e } (\tau) \delta e\right)d_{\omega}e + 
    p'_{W} W_1^{-1}(\tau \delta e)d_{\omega}e - p'_{X} \widetilde{\rho}^{-1}\left(\frac{\delta \widetilde{\rho}}{\delta e } (\tau) \delta e\right)d_{\omega}e
\end{align*}
as shown in Lemma \ref{lem:expression_deltatau_constrained} where $p'_{X}$ is the projection to the intersection of the complement of the kernel of $\widetilde{\rho}$ and $W_1^{\partial, (N-3,N-1)}$. Using this last computation we can compute the variation of the Hamiltonian constraint $H_{\lambda}$:

\begin{align*}
    \delta H_{\lambda} &= \int_{\Sigma} \lambda e_n  e^{N-4}\delta e  F_{\omega}+\frac{1}{(N-2)!}\Lambda  \lambda e_n e^{N-2} \delta e -\frac{1}{(N-3)}\lambda e_n e^{N-3}  d_{\omega} \delta \omega \\
     & \quad - \lambda p_{\mathcal{S}} (e_n e^{N-4}\delta \omega) d_{\omega}e - (N-4)\lambda p_{\mathcal{S}} (e_n e^{N-5}\delta e (\omega-\omega_0)) d_{\omega}e \\
    & \quad -  \delta_e \tau' d_{\omega}e + \tau' [\delta \omega, e] -\tau' d_{\omega}\delta e\\ 
&= \int_{\Sigma} \lambda e_n  e^{N-4}\delta e  F_{\omega} +\frac{1}{(N-2)!}\Lambda  \lambda e_n e^{N-2} \delta e+  \frac{1}{(N-3)} d_{\omega}(\lambda e_n) e^{N-3}  \delta \omega \\
    & \quad + \lambda e_n e^{N-4}d_{\omega} e  \delta \omega - \lambda  e_n e^{N-4} \delta \omega p_{\mathcal{T}}(d_{\omega}e)\\
    & \quad- (N-4)\lambda e_n e^{N-5}\delta e (\omega-\omega_0) p_{\mathcal{T}}(d_{\omega}e)  -  \delta e g(\tau',\omega ,e) + \tau' [\delta \omega, e] -\tau' d_{\omega}\delta e\\
&= \int_{\Sigma} \lambda e_n  e^{N-4}\delta e  F_{\omega} +\frac{1}{(N-2)!}\Lambda  \lambda e_n e^{N-2} \delta e+  \frac{1}{(N-3)} d_{\omega}(\lambda e_n) e^{N-3}  \delta \omega\\
    & \quad + \lambda \sigma  e^{N-3} \delta \omega- (N-4)\lambda e_n e^{N-5}\delta e (\omega-\omega_0) p_{\mathcal{T}}(d_{\omega}e) \\
    & \quad-  \delta e g(\tau',\omega ,e) + \tau' [\delta \omega, e] -\tau' d_{\omega}\delta e
\end{align*}
where $\tau'=  p_{\mathcal{S}}(\lambda e_n e^{N-4} (\omega-\omega_0))$ and we used \eqref{e:structural_constraint_mod}. From the expressions of the variation of the constraints we can deduce their Hamiltonian vector fields. Let $X$ be a generic constraint, then we denote with $\mathbb{X}$ the corresponding Hamiltonian vector field $\iota_{\mathbb{X}}\varpi^{\partial}_{PC}= \delta X$ and with $\mathbb{X}_e$ $\mathbb{X}_{\omega}$ its components, i.e.
\begin{align*}
    \mathbb{X} = \mathbb{X}_e \frac{\delta}{\delta e} + \mathbb{X}_{\omega} \frac{\delta}{\delta \omega}.
\end{align*}
Hence we have
\begin{align*}
&\begin{aligned}
    &\mathbb{L}_e = [c,e] &  \qquad \qquad \qquad\qquad
    & \mathbb{L}_\omega = d_{\omega} c \\ 
    &\mathbb{P}_e = - \mathcal{L}_{\xi}^{\omega_0} e & 
    & \mathbb{P}_\omega = - \mathcal{L}_{\xi}^{\omega_0} (\omega-\omega_0) - \iota_ {\xi}F_{\omega_0}\\
    & e^{N-3} \mathbb{R}_e= [\tau, e] & 
    & e^{N-3} \mathbb{R}_\omega = g(\tau,\omega ,e) + d_{\omega} \tau 
\end{aligned}\\    
    & e^{N-3} \mathbb{H}_e = \frac{1}{(N-3)} e^{N-3} d_{\omega}(\lambda e_n) + \lambda e^{N-3} \sigma - [\tau', e]\\
    & e^{N-3} \mathbb{H}_\omega  =  \lambda e_n e^{N-4} F_{\omega}+\frac{1}{(N-2)!}\Lambda  \lambda e_n e^{N-2}- (N-4)\lambda e_n e^{N-5}(\omega-\omega_0)p_{\mathcal{T}}(d_{\omega}e)\\
    & \qquad\qquad - g(\tau',\omega ,e) - d_{\omega} \tau'.
\end{align*}
 The components $\mathbb{R}_\omega$ and $\mathbb{H}_\omega$ are uniquely determined requiring the structural constraint \eqref{e:structural_constraint_mod}. The components $\mathbb{R}_e$ and $\mathbb{H}_e$ are recovered by inversion of $W_{N-3}^{\partial, (1,1)}$ (which is possible thanks to Lemma \ref{lem:[tau,e]inImW}). Following these we compute the Poisson brackets between the constraints and analyse their structure. The brackets between $L_c$ and $P_{\xi}$ are the same as in the non-degenerate case presented in \cite{CCS2020}:
\begin{align*}
    \{L_c, L_c\}  = - \frac{1}{2} L_{[c,c]}; \quad
    \{L_c, P_{\xi}\} =  L_{\mathcal{L}_{\xi}^{\omega_0}c}; \quad
    \{P_{\xi}, P_{\xi}\}  = \frac{1}{2}P_{[\xi, \xi]} - \frac{1}{2}L_{\iota_{\xi}\iota_{\xi}F_{\omega_0}}.
\end{align*}
Let us now compute the brackets between $L_c$, $P_{\xi}$ and $R_{\tau}$. In both computations we use the results of Lemmas \ref{lem:expression_deltatau_constrained} and \ref{lem:computations_for_LR_and_PR} and the properties of $\tau$.
\begin{align*}
    \{L_c, R_{\tau}\}= \int_{\Sigma} & [c,e] g(\tau,\omega ,e) + [c,e]d_{\omega} \tau + d_{\omega} c [\tau, e] \\
        = \int_{\Sigma} & [c,e] g(\tau,\omega ,e)- [c, \tau] d_{\omega} e \\
        = \int_{\Sigma} & p'_{\mathcal{S}}[c, \tau] d_{\omega} e- [c, \tau] d_{\omega} e = \int_{\Sigma} - p_{\mathcal{S}}[c, \tau] d_{\omega} e = -R_{p_{\mathcal{S}}[c, \tau]};
 \end{align*}       
\begin{align*}
    \{P_{\xi},R_{\tau}\}= \int_{\Sigma} & - [\tau, e]\mathcal{L}_{\xi}^{\omega_0} (\omega-\omega_0) - [\tau, e]\iota_ {\xi}F_{\omega_0} - \mathcal{L}_{\xi}^{\omega_0} e g(\tau,\omega ,e) - \mathcal{L}_{\xi}^{\omega_0} e d_{\omega} \tau \\
        = \int_{\Sigma} & - \mathcal{L}_{\xi}^{\omega_0} e g(\tau,\omega ,e) + \mathcal{L}_{\xi}^{\omega_0}\tau d_{\omega} e\\
        = \int_{\Sigma} & -p'_{\mathcal{S}}\mathcal{L}_{\xi}^{\omega_0}\tau d_{\omega} e+ \mathcal{L}_{\xi}^{\omega_0}\tau d_{\omega} e = \int_{\Sigma}  p_{\mathcal{S}}\mathcal{L}_{\xi}^{\omega_0}\tau d_{\omega} e = R_{p_{\mathcal{S}}\mathcal{L}_{\xi}^{\omega_0}\tau}.
\end{align*}
We now compute the brackets between $L_c$, $P_{\xi}$ and $H_{\lambda}$.
\begin{align*}
    \{L_c,  H_{\lambda}\}  = & \int_{\Sigma}  [c,e] e^{N-4} \lambda e_n F_{\omega} +\frac{1}{(N-2)!}[c,e]\Lambda  \lambda e_n e^{N-2} -[c,e]g(\tau',\omega ,e) \\ 
    & \quad -[c,e]d_{\omega} \tau'
     - (N-4) [c,e] \lambda e_n e^{N-5}(\omega-\omega_0)p_{\mathcal{T}}(d_{\omega}e) \\ 
    & \quad + \frac{1}{(N-3)} e^{N-3} d_{\omega} c  d_{\omega}(\lambda e_n) +  e^{N-3} d_{\omega} c\lambda \sigma - d_{\omega} c [\tau', e]\\ 
        =  & \int_{\Sigma}  - \frac{1}{(N-3)} [c, \lambda e_n ] e^{N-3} F_{\omega}-\frac{1}{(N-1)!}\Lambda[c, \lambda e_n ] e^{N-1} \\ 
        & \quad+ p_{\mathcal{S}}([c, \tau'] - \lambda e_n e^{N-4}d_{\omega} c - [c,e^{N-4}] \lambda e_n (\omega-\omega_0))d_{\omega}e\\
        = & \int_{\Sigma}  - \frac{1}{(N-3)}\left([c, \lambda e_n ]^{(a)}e_a e^{N-3} F_{\omega} -[c, \lambda e_n ]^{(n)}e_n e^{N-3} F_{\omega}\right)\\
        & \quad -\frac{1}{(N-1)!}\Lambda[c, \lambda e_n ]^{(n)}e_n e^{N-1} -  p_{\mathcal{S}}(\lambda e_n e^{N-4} d_{\omega_0} c ) d_{\omega} e\\
        & \quad
         +p_{\mathcal{S}}([c, \lambda e_n ]^{(a)}e_a e^{N-4}(\omega- \omega_0) +[c, \lambda e_n ]^{(n)}e_n e^{N-4}(\omega- \omega_0) ) d_{\omega} e  \\
         = & - P_{[c, \lambda e_n ]^{(a)}} + L_{[c, \lambda e_n ]^{(a)}(\omega - \omega_0)_a} - H_{[c, \lambda e_n ]^{(n)}}\\
         & + R_{p_{\mathcal{S}}([c, \lambda e_n ]^{(a)}e_a e^{N-4}(\omega- \omega_0))} - R_{p_{\mathcal{S}}(\lambda e_n e^{N-4}d_{\omega_0} c)};
\end{align*}
\begin{align*}
    \{P_{\xi},H_{\lambda}\}  = & \int_{\Sigma} - \mathcal{L}_{\xi}^{\omega_0} e \lambda e_n e^{N-4}F_{\omega} -\frac{1}{(N-2)!}\Lambda \mathcal{L}_{\xi}^{\omega_0} e \lambda e_n e^{N-2}+ \mathcal{L}_{\xi}^{\omega_0} e g(\tau',\omega ,e)\\
        & \quad  + \mathcal{L}_{\xi}^{\omega_0} e d_{\omega} \tau'+ (N-4) \mathcal{L}_{\xi}^{\omega_0} e \lambda e_n e^{N-5}(\omega-\omega_0)p_{\mathcal{T}}(d_{\omega}e)\\
    & \quad- \left( \mathcal{L}_{\xi}^{\omega_0} (\omega-\omega_0)+ \iota_ {\xi}F_{\omega_0}\right)\left(\frac{e^{N-3} d_{\omega}(\lambda e_n)}{N-3} + \lambda e^{N-3} \sigma - [\tau', e] \right) \\
    = & \int_{\Sigma} \frac{1}{(N-3)}  \mathcal{L}_{\xi}^{\omega_0} (\lambda e_n) e^{N-3} F_{\omega} +\frac{1}{(N-1)!}\Lambda  e^{N-1} \mathcal{L}_{\xi}^{\omega_0}(\lambda e_n) \\
    & \quad + p_{\mathcal{S}}\left( -\mathcal{L}_{\xi}^{\omega_0} \tau '+ \lambda e_n e^{N-4}\left( \mathcal{L}_{\xi}^{\omega_0} (\omega-\omega_0)+ \iota_ {\xi}F_{\omega_0}\right)\right)d_{\omega} e \\
    & \quad + p_{\mathcal{S}}(\mathcal{L}_{\xi}^{\omega_0}( e^{N-4}) \lambda e_n (\omega-\omega_0))d_{\omega}e\\
     =& \int_{\Sigma}  \frac{1}{(N-3)} \left(\mathcal{L}_{\xi}^{\omega_0} (\lambda e_n)^{(a)}e_a  e^{N-3} F_{\omega}+\mathcal{L}_{\xi}^{\omega_0} (\lambda e_n)^{(n)}e_n e^{N-3} F_{\omega} \right) \\
     & \quad +\frac{1}{(N-1)!}\Lambda  e^{N-1} \mathcal{L}_{\xi}^{\omega_0}(\lambda e_n)^{(n)}e_n + p_{\mathcal{S}} (\lambda e_n e^{N-4}\iota_ {\xi}F_{\omega_0})d_{\omega} e\\
    & - p_{\mathcal{S}}\left( \mathcal{L}_{\xi}^{\omega_0}(\lambda e_n)^{(n)} e_n e^{N-4}(\omega- \omega_0)+ \mathcal{L}_{\xi}^{\omega_0}(\lambda e_n)^{(a)} e_a e^{N-4}(\omega- \omega_0)  \right)d_{\omega} e \\
     =  & P_{ \mathcal{L}_{\xi}^{\omega_0} (\lambda e_n)^{(a)}} +H_{ \mathcal{L}_{\xi}^{\omega_0} (\lambda e_n)^{(n)}}-L_{ \mathcal{L}_{\xi}^{\omega_0} (\lambda e_n)^{(a)} (\omega - \omega_0)_a}\\
     & -R_{p_{\mathcal{S}}(\mathcal{L}_{\xi}^{\omega_0}(\lambda e_n)^{(a)} e_a e^{N-4}(\omega- \omega_0))} + R_{p_{\mathcal{S}}(\lambda e_n e^{N-4}\iota_ {\xi}F_{\omega_0})}.
\end{align*}
We now compute the remaining brackets $\{R_{\tau},R_{\tau}\}$, $\{R_{\tau},H_{\lambda}\}$ and  $\{H_{\lambda},H_{\lambda}\}$.
Since $H_{\lambda}$ contains terms proportional to $R_{\tau}$ (for $\tau= p_{\mathcal{S}}( \lambda e_n e^{N-4} (\omega- \omega_0)$) we first compute the brackets between two $R_{\tau}$ and then the others:
\begin{align*}
    \{R_{\tau},R_{\tau}\}= \int_{\Sigma} & W_{N-3}^{-1}([\tau, e])g(\tau,\omega ,e) +  W_{N-3}^{-1}([\tau, e]) d_{\omega} \tau . 
\end{align*}
The first term is proportional to $d_{\omega}e$ by construction, so it will be 0 on shell. Let us concentrate on the second term. We want to prove, using normal geodesic coordinates, that it is not proportional to any of the constraints and not 0. Let us fix a point $p \in \Sigma$ and consider an open neighbourhood $U$ of it. From Proposition \ref{prop:components_of_tau} we deduce that the unique components at the point $p$ with respect to the standard basis that compose $\tau$ are $X_{\mu_2}^{\mu_1}, Y_{\mu}$ for $\mu, \mu_1, \mu_2=1 \dots N-2$ subject to 
\begin{align*}
    \sum_{\mu=1}^{N-2} Y_{\mu} =0 \text{ and } 
    X_{\mu_1}^{\mu_2} =- X_{\mu_2}^{\mu_1}.
\end{align*}
The first equation holds also on the whole neighborhood while the second set holds only on the point $p$. From Corollary \ref{cor:components_of_W-1[tau,e]} we know that the non zero components in  $W_{N-3}^{-1} ([\tau, e])$ are
\begin{align*}
        [W_{N-3}^{-1}([\tau,e])]_{\mu_1}^{\mu_2} & \propto X_{\mu_1}^{\mu_2} \\
        [W_{N-3}^{-1}([\tau,e])]_{\mu}^{\mu} & \propto Y_{\mu} \\
    \end{align*}
    such that $\sum_{\mu=1}^{N-2} [W_{N-3}^{-1}([\tau,e])]_{\mu}^{\mu}=0 $ and $[W_{N-3}^{-1}([\tau,e])]_{\mu_1}^{\mu_2}=-[W_{N-3}^{-1}([\tau,e])]_{\mu_2}^{\mu_1}$.
    
Furthermore, from Proposition \ref{prop:components_of_tau} we also know that the non zero components of $\tau$ are 
$Y_{\mu}$ and $ X_{\mu_1}^{\mu_2}$
such that
\begin{align*}
    \sum_{\mu=1}^{N-2} Y_{\mu} =0 \text{ and } 
    X_{\mu_1}^{\mu_2} = f(\widetilde{g}^{\partial}, X_{\mu_2}^{\mu_1}, Y_{\mu})
\end{align*}
for $\mu_1 < \mu_2$  and some linear function $f$.  Remembering that $W_{N-3}^{-1}([\tau, e]) d_{\omega} \tau$ should  be a volume form, we deduce that, on shell,
\begin{align*}
    W_{N-3}^{-1}([\tau, e]) d_{\omega} \tau & = \left(
    \sum_{\mu=1}^{N-2} Y_{\mu} \partial_{N-1} Y_{\mu} +
    \sum_{\mu_1, \mu_2=1}^{N-2} X_{\mu_1}^{\mu_2}\partial_{N-1} X_{\mu_2}^{\mu_1}
    \right)\mathbf{V} \\
    & = \left(
    \sum_{\mu=1}^{N-2} Y_{\mu} \partial_{N-1} Y_{\mu} +
    \sum_{\mu_1 < \mu_2 \mu_1,\mu_2=1}^{N-2} X_{\mu_1}^{\mu_2}\partial_{N-1} f(\widetilde{g}^{\partial}, X_{\mu_1}^{\mu_2}, Y_{\mu})
    \right)\mathbf{V}\\
    &= \colon F_{\tau \tau}
\end{align*}
where $\mathbf{V}= e_{1} \dots e_{N-1} e_n dx^{1} \dots dx^{N-1} $.
This quantity is for generic  $\tau$ different from zero, on shell. Hence 
\begin{align*}
    \{ R_{\tau}, R_{\tau} \} \approx F_{\tau \tau} \not\approx 0.
\end{align*}

With this result we can more easily compute the last two brackets:
\begin{align*}
    \{H_{\lambda},  H_{\lambda}\}= \int_{\Sigma} & \left( \frac{1}{(N-3)} d_{\omega}(\lambda e_n)+\lambda  \sigma- W_{N-3}^{-1}([\tau', e])\right)\lambda e_n e^{N-4} F_{\omega}\\
    + & \left( \frac{1}{(N-3)} d_{\omega}(\lambda e_n)+\lambda  \sigma- W_{N-3}^{-1}([\tau', e])\right)\frac{1}{(N-2)!}\Lambda  \lambda e_n e^{N-2} \\
    - & \left( \frac{1}{(N-3)} d_{\omega}(\lambda e_n)+\lambda  \sigma- W_{N-3}^{-1}([\tau', e])\right)g(\tau',\omega ,e) \\
    - & \left( \frac{1}{(N-3)} d_{\omega}(\lambda e_n)+\lambda  \sigma- W_{N-3}^{-1}([\tau', e])\right)d_{\omega} \tau'\\
    -&(N-4)\frac{1}{(N-3)} d_{\omega}(\lambda e_n)\lambda e_n e^{N-5}(\omega-\omega_0)p_{\mathcal{T}}(d_{\omega}e)\\
    -&(N-4) \left(\lambda  \sigma  - W_{N-3}^{-1}([\tau', e])\right)\lambda e_n e^{N-5}(\omega-\omega_0)p_{\mathcal{T}}(d_{\omega}e).
\end{align*}
Since $\lambda$ and $e_n$ are odd quantities and $\tau'= \lambda p_{\mathcal{S}}(e_n e^{N-4}(\omega-\omega_0))$, the terms in the first two lines  and in the last two vanish. Furthermore the last terms of the third and fourth lines are the one composing the brackets $\{R_{\tau '},R_{\tau '}\}$. Expanding the first and the second term of the third line we get
\begin{align*}
     \widetilde{\rho}^{-1}([\tau', d_{\omega}(\lambda e_n)])d_{\omega}e + W_1^{-1}(\tau' d_{\omega}(\lambda e_n))d_{\omega}e +  \widetilde{\rho}^{-1}([\tau', \lambda  \sigma])d_{\omega}e + W_1^{-1}(\tau' \lambda  \sigma)d_{\omega}e.
\end{align*}
All these terms are zero since they encompass terms with either $\lambda\lambda=0$ or $ e_n e_n=0$. We can draw the same conclusion also for the following term:
\begin{align*}
    d_{\omega}(\lambda e_n)d_{\omega} \tau'= [ F_{\omega}, \lambda e_n] \tau'=0.
\end{align*}
The same holds also for the term $\lambda  \sigma d_{\omega} \tau'$ since both $\sigma$ and $\tau'$ contain $e_n$.\footnote{ Using the lemmas in Section \ref{sec:technical} it is possible to prove that all the non-zero components of $\sigma$ are in the direction of $e_n$.} Hence
\begin{align*}
    \{H_{\lambda},  H_{\lambda}\}=\{R_{\tau '},R_{\tau '}\} \approx F_{\tau' \tau'} \not \approx 0.
\end{align*}

The last bracket that we have to compute is $\{ R_{\tau}, H_{\lambda} \}$. From the expression of the Hamiltonian vector fields we get
\begin{align*}
    \{ R_{\tau}, H_{\lambda} \}= \int_{\Sigma} & \frac{1}{(N-3)} d_{\omega}(\lambda e_n) \left( g(\tau,\omega ,e)+d_{\omega} \tau \right) + \lambda \sigma g(\tau,\omega ,e)+ \lambda \sigma d_{\omega} \tau \\
    + &  W_{N-3}^{-1}([\tau, e])\lambda e_n e^{N-4} F_{\omega}-  W_{N-3}^{-1}([\tau', e])\left( g(\tau,\omega ,e)+d_{\omega} \tau \right)\\
    + & \frac{1}{(N-2)!}\Lambda  \lambda e_n e[\tau, e] - W_{N-3}^{-1}([\tau, e])\left( g(\tau',\omega ,e)+d_{\omega} \tau ' \right) \\
    - &  (N-4)W_{N-3}^{-1}([\tau, e])\lambda e_n e^{N-5}(\omega-\omega_0)p_{\mathcal{T}}(d_{\omega}e).
\end{align*}
The last two terms of the second and third lines are the one composing the brackets $\{R_{\tau },R_{\tau '}\}$, and the first term of the third line vanishes because $e \tau =0$ and $[e,e]=0$.
We want to prove that $\{ R_{\tau}, H_{\lambda} \} \not\approx 0$. 
Using coordinate expansion one can prove that the second and the fifth term have the same expression and read:
\begin{align*}
    d_{\omega}(\lambda e_n)d_{\omega}\tau &+ W_1^{-1}([\tau, e])\lambda e_nF_{\omega}
    =  - [F_{\omega} , \lambda e_n] \tau + W_1^{-1}([\tau, e])\lambda e_nF_{\omega} \\
    &=2 \lambda \sum_{\mu=1}^{N-2} Y_{\mu} (F_{\omega})_{\mu N-1 }^{\mu N-1} + \lambda \sum_{\mu_1, \mu_2=1}^{N-2} X_{\mu_1}^{\mu_2}(F_{\omega})_{\mu_2 N-1 }^{\mu_1 N-1} = \colon G_{\lambda \tau}.
\end{align*}
These terms are not proportional to any of the constraints and not proportional to $\{R_{\tau },R_{\tau '}\}$. The term in the fourth line is proportional to $R_{\tau}$ so we can discard it.
Let us now consider the fourth term:  since $d_{\omega} \tau$ is in the image of $W_1$ we can invert it and get 
\begin{align*}
    \lambda \sigma d_{\omega} \tau &= \lambda e^{N-3} \sigma W^{-1}(d_{\omega} \tau)\\
    &= \lambda e_n e^{N-4}d_{\omega} e W^{-1}(d_{\omega} \tau) - \lambda e_n e^{N-4} p_{\mathcal{T}}(d_{\omega} e ) W_1^{-1}(d_{\omega} \tau).
\end{align*}
The second term is again proportional to $R_{\tau}$ so we can discard it as well.  Let us now consider the first term of this expression and $d_{\omega}(\lambda e_n)g(\tau,\omega ,e) + \lambda \sigma g(\tau,\omega ,e)$ --- the last two remaining terms. By expanding these terms using the definition of $f$, integrating by parts and using $\tau \wedge e_n = 0$ we get that these three terms add up    to zero.
 Collecting these results we get
\begin{align*}
    \{ R_{\tau}, H_{\lambda} \}\approx  \{R_{\tau },R_{\tau '}\} + G_{\lambda \tau} \approx F_{\tau \tau'} + G_{\lambda \tau} \not \approx 0 .
\end{align*}
\end{proof}

\begin{remark}
For $N=4$ some of the previous computation simplify. In particular it is possible to give a compact explicit expression for the function $F_{\tau \tau}$. This coincides with the corresponding one of the linearized theory $\widetilde{F}_{\tau \tau}$ expressed in \eqref{e:RR_lin_4}. As a consequence it is also possible to give an explicit expression for the other brackets not proportional to the constraints.
\end{remark}

{
\begin{remark} \label{rmk:zeromodes_notlin}
As we will see in Appendix \ref{sec:linearized_theory}, in the linearized case we can identify some first class zero modes inside the second class constraint (see Remark \ref{rmk:zeromodes_linearized}). In the non-linearized case such identification is more complicated but such modes should anyway be present. This will be object of future studies.
\end{remark}
}

\begin{corollary}\label{cor:constraints-first-second_class}
    The constraints $L_c$, $P_{\xi}$, $H_{\lambda}$ and $R_{\tau}$ do not form a first class system. In particular $R_{\tau}$ is a second class constraint while the others are first class (as defined in Remark \ref{rem:first-second_class_constraints}).
\end{corollary}
\begin{proof}
 Throughout the proof we use the notation and terminology established in Appendix \ref{sec:first-second_class_def}.
Since the bracket between $R_{\tau}$ and itself is not zero on shell the system contains constraints that are second class. We want now to establish which constraints are of second class and which are of first class. The constraints $L_c$ and $P_{\xi}$ commute ---on shell--- with themselves and all the other constraints, hence they are of first class. Let us now consider $R_{\tau}$ and $H_{\lambda}$. We want to prove that $R_{\tau}$ is of second class while, using a linear transformation of the constraints $H_{\lambda}$ is of first class. Using the result of Proposition \ref{prop:numberofsecondclassconstraints}, if we call $D$ the matrix representing the bracket $\{R_{\tau },R_{\tau }\}$, $B$ the one representing the bracket  $\{ R_{\tau}, H_{\lambda} \}$, and $C$ the one  representing the bracket $ \{H_{\lambda},  H_{\lambda}\}$, we have to prove that {$B^T D^{-1} B=-C$}. 

From the proof of Theorem  \ref{thm:Brackets_constraints} we can deduce the expressions of the matrices $B$, $D$ and $C$. All the  components of such matrices contain a derivative in the \textit{lightlike} direction, apart from the terms coming from $G_{\lambda \tau}$ in $B$. Hence all components of $D^{-1}$ will contain the inverse of such derivative. Since $\lambda$ is an odd quantity, all the terms contained in {$B^T D^{-1} B$} without a derivative vanish because of Lemma \ref{lem:vanishofBD-1BT}. Hence the only surviving elements in {$B^T D^{-1} B$} come from the multiplication of the elements containing a derivative in $B$. We denote such terms by $B'$. It is then a straightforward computation to check that the coefficients of such combination are actually equal to those of $C$. Indeed, since these matrices have the same functional form ($F_{\tau \tau}$), we can express the matrices $B'$ and $C$ respectively as
{$B'=   D p_{\mathcal{S}}(e_n e^{N-4} (\omega-\omega_0))$} and $C=p_{\mathcal{S}}(e_n e^{N-4} (\omega-\omega_0))^T D p_{\mathcal{S}}(e_n e^{N-4} (\omega-\omega_0))$. Hence we have
{
\begin{align*}
    B'^T D^{-1} B' & = p_{\mathcal{S}}(e_n e^{N-4} (\omega-\omega_0))^T D^T D^{-1} D p_{\mathcal{S}}(e_n e^{N-4} (\omega-\omega_0)) \\
    & = -p_{\mathcal{S}}(e_n e^{N-4} (\omega-\omega_0))^T D p_{\mathcal{S}}(e_n e^{N-4} (\omega-\omega_0))=-C.
\end{align*}}
\end{proof}

We can now count the degrees of freedom of the reduced phase space. From the definition given in Section \ref{sec:first-second_class_def} we can deduce that the correct number of physical degrees of freedom is given by \cite[(1.60)]{HT}: let $r$ be the number of degrees of freedom  of the reduced phase space, $p$ the number of degrees of freedom of the geometric phase space, $f$ the number of first class constraints and $s$ the number of second class constraints, then
\begin{align*}
    r = p - 2f - s.
\end{align*}
In our case these quantities have the following values:
the geometric phase space has $2 N (N-1)$ degrees of freedom. From Corollary \ref{cor:constraints-first-second_class} we have that there are  $\frac{N(N-1)}{2}+ N= \frac{N(N+1)}{2}$ first class constraints and $\frac{N(N-3)}{2}$ second class constraints (see Proposition \ref{prop:components_of_tau} for the number of degrees of freedom of $\tau$). We can deduce that the correct number of local degrees of freedom is given by
  
\begin{align*}
    2 N (N-1) - N(N+1) - \frac{N(N-3)}{2} = \frac{N(N-3)}{2}.
\end{align*}
 In the case $N=4$ this computation produces two local degrees of freedom. This result agrees with the previous works in the literature (e.g. \cite{AlexandrovSpeziale15}).

\appendix
\section{First and Second class constraints}\label{sec:first-second_class_def}
An important distinction between the constraints of a system is the one provided by the difference between first and second class constraints. In this section we review the definition and prove a result to easily distinguish the two classes. 

Roughly speaking, a constraint is of second class  if its Poisson brackets with other constraints do not vanish  on the constrained surfaces. However, this definition is not precise since it is always possible to take linear combinations of the constraints without modifying the reduced phase space of the theory.  Furthermore first and second class constraints correspond to different physical interpretations: the first ones are in one to one correspondence with the generators of gauge transformations of the theory, while the second ones are just identities through which we can express some canonical variables in terms of the other. Hence, to correctly encompass these differences, we need a more sophisticated definition. Starting from the results presented in \cite[Chapter 1]{HT} we can give the following definition:
\begin{definition}
    Let $\mathcal{F}$ be a symplectic manifold and let $\phi_i \in C^{\infty}(\mathcal{F})$ be a set of smooth functions on it. Denote with $C_{ij}=\{\phi_i, \phi_j\}$ the matrix of the Poisson brackets of the functions. Then the number of second class functions of the set is the rank\footnote{We assume the rank to be constant on the zero locus.} of the matrix $C_{ij}$ on the zero locus of the  functions. In particular if  $C_{ij} \approx 0 $ then we say that all the functions are first class.
\end{definition}
This definition clearly coincides with the standard one in case all the constraints are first class, i.e. all the  constraints commute with every other one. However, it allows us to treat the general case, since it is invariant under rearranging the constraints by linear combinations.
We now state a result that will be helpful in assessing the number of second class constraints in a system.
\begin{proposition}\label{prop:numberofsecondclassconstraints}
    Let $\mathcal{F}$ be a symplectic manifold and let $\psi_i,\phi_j \in C^{\infty}(\mathcal{F})$, $i=1\dots n$, $j=1\dots m$. Denote with $C_{jj'}, B_{ij}, D_{i i'}$ respectively the matrices representing the Poisson brackets  $\{\phi_j,\phi_{j'}\}$, $\{\psi_i,\phi_j\}$ and $\{\psi_i,\psi_{i'}\}$, with $i,i'=1\dots n$, $j,j'=1\dots m$. Then, if $D$ is invertible and {$C= -B^T D^{-1}B$}, the number of second class constraints is $n$, i.e. the rank of the matrix $D$.   
\end{proposition}

\begin{remark}\label{rem:first-second_class_constraints}
In this case, we will say that the $\phi$'s are the first class constraints and the $\psi$'s the second class constraints of the system.
\end{remark}

\begin{proof}
The matrix representing the Poisson brackets has the form
\begin{align*}  
    P= \left( 
    \begin{array}{cc}
        C & -B^T \\
        B & D
    \end{array}
    \right)
\end{align*}
where the blocks are as in the statement. We want to prove that this matrix is congruent to one of rank $n$ i.e. that there exists an invertible matrix $Q$ such that $Q^{T}PQ$ has rank $n$. Since $D$ is invertible, we can build $Q$ as follows:
\begin{align*}
    Q= \left( 
    \begin{array}{cc}
        1 &  0 \\
        -D^{-1} B & 1
    \end{array}
    \right).
\end{align*}
An easy computation shows that
\begin{align*} 
    Q^T P Q= \left( 
    \begin{array}{cc}
        C+ B^T D^{-1} B & 0 \\
        0 & D
    \end{array}
    \right).
\end{align*}
Hence, using the second hypothesis {$C=- B^T D^{-1}B$} we get the claim.
\end{proof}

\begin{remark}\label{rem:linearcomb_constraints_fsc}
This result shows explicitly that a naive definition of first class constraint as the one commuting with everything else is not sufficient to correctly consider more involved cases where the constraints do not commute (under the Poisson brackets) on the nose, but there are linear combinations of them that do. In this specific setting, from the proof of the Proposition, we gather that we can consider the set of functions
\begin{align*}
    \widetilde{\phi}_j = \phi_j + \sum_{i,i'} B_{ij}D^{i i'} \psi_{i'}; \qquad \widetilde{\psi}_i = \psi_i
\end{align*}
and conclude that the functions $\widetilde{\phi}_j$ are first class (in the classical sense) and $\widetilde{\psi}_i$ are second class.
\end{remark}

\section{Linearized theory} \label{sec:linearized_theory}

In this section we analyse the boundary structure of the linearized theory. In particular we first introduce it on the bulk and then construct the boundary theory respectively in the non-degenerate and degenerate case.  We present the results only in the case $N=4$.  We denote with a tilde the linearized quantities to distinguish them from the general ones {and use the same notation introduced in Section \ref{sec:technical}. The unique difference is that we will denote by $W_{e_0}^{\bullet}$ the maps $e_0 \wedge \cdot$ to highlight the difference with the normal case.} The results of this appendix overlap with \cite{T2019b}.

\subsection{Linearized field equations and boundary structure}
Consider the action \eqref{e:PCaction} of the Palatini-Cartan theory with the following choices of coframe and connection:
\begin{align*}
 e&=e_0+b,\\
\omega&=\omega_0+a   
\end{align*}
with $(e_0,\omega_0)$ a fixed solution of Euler--Lagrange equations \eqref{e:tf} and \eqref{e:ee}  of the standard Palatini--Cartan theory. We retain only the quadratic terms in $a,b$; thus:
\begin{equation*}
S_{LPC}=\int_{M}{\left(\frac{1}{2}b b  F_{\omega_0}+e_0 b d_{\omega_0}a+\frac{1}{4}e_0 e_0 [a,a]{+ \frac{1}{4} \Lambda e_0e_0 bb}\right)}.
\end{equation*}
This produces the following Euler--Lagrange equations:
\begin{subequations}
\label{e:linearel}
\begin{align}
e_0(d_{\omega_0}b+[a,e_0])&=0 \label{eq:linearel1}\\
b\Fo+e_0 d_{\omega_0}a{+ \frac{1}{2} \Lambda e_0e_0 b}&=0\label{eq:linearel2}.
\end{align}
\end{subequations}
The first equation, as in the non-linearized case is equivalent to
$d_{\omega_0}b+[a,e_0]=0$.

With the same procedure derived from \cite{KT1979} used for the general theory, we can construct the \emph{geometric phase space} also for the linearized theory. The steps are exactly the same, while in this last case the kernel is parametrized by vector fields  {$X= v \frac{\delta}{\delta \omega}$} with $v$ satisfying 
\begin{equation}\label{e:Xomega_lin}
e_0 v=0
\end{equation} 
instead of \eqref{e:Xomega}. Consequently the \emph{geometric} space of boundary fields of the linearized theory,  $\widetilde{\mathcal{F}}_{LPC}$
is then parametrized by the field $b$ and by the equivalence classes of $a$ under the relation $a \sim a + v$ with $v$ satisfying \eqref{e:Xomega_lin}. The symplectic form on  $\widetilde{\mathcal{F}}_{LPC}$ is given by
\begin{align}\label{e:classical-boundary-symplform_lin}
\widetilde{\varpi}_{LPC} = \int_{\Sigma} e_0 \delta b \delta [a].
\end{align}

The following proposition provides a shortcut (possible only in the linearized case) to the choice of a representative of the equivalence class:
\begin{proposition}\label{prop:defTheta}
There exists a symplectomorphism $\widetilde{\mathcal{F}}_{LPC} \to T^*\Omega^{1,1}_{\partial}$ equipped with the canonical symplectic form.
\end{proposition}
\begin{proof}
Let $b$ and $\Theta$ be the fields respectively in the base and fiber of $\mathcal{F}^\partial_{LPC}=T^*\Omega^{1,1}_{\partial}$. The symplectic form of this space is 
\begin{equation*}
\varpi^\partial_{LPC}=\int_{\Sigma}{\delta b \delta \Theta}.
\end{equation*}
From Lemma \ref{lem:We_boundary} we know that the map $W_{e_0}^{\partial,(1,2)}$ is surjective but not injective.
Hence $\Theta$ can be written as 
\begin{equation} \label{e:deftheta}
\Theta = e_0 a 
\end{equation}
for some $a \in \Omega^{1,2}_{\partial}$. Because of the definition of $[a]$, it is then clear that there is a bijection between $[a]$ and $\Theta$. This bijection is also a symplectomorphism since it sends the symplectic form \eqref{e:classical-boundary-symplform_lin} to the corresponding one of $\mathcal{F}^\partial_{LPC}$. 
\end{proof}

\begin{remark}
This symplectomorphism exists also in the non-linearized case, but in both the degenerate and non-degenerate case it is not possible to write the action in a simple way with the new variables. Hence this is an important feature of the linearized case.
\end{remark}

\subsection{Non-degenerate boundary metric}
In this section we will implement the results of the non-degenerate  theory to the linearized case. Therefore we will consider a background boundary coframe giving rise to a non-degenerate boundary metric $g^\partial_0$. Moreover, we will compute the algebra of constraints, concluding that the reduced phase space of the linearized theory is coisotropic.

In this setting, the constraints of the theory are given by 
\begin{align*}
  e_0(\Do b+[a,e_0])=0 \quad\text{and}\quad b\Fo+e_0\Do a{+ \frac{1}{2} \Lambda e_0e_0 b}=0.  
\end{align*}
 Hence, using the identification \eqref{e:deftheta} we can write the constraints of the boundary linearized theory as in the following definition. Let now $c \in\Omega^{0,2}_\partial[1]$,  and $\mu \in \Omega^{0,1}_\partial[1]$ be (odd) Lagrange multipliers where the notation $[1]$ denotes that the fields are shifted by 1 and are treated as odd variables.

The functionals defining the constraints of the non-degenerate linearized Palatini-Cartan theory are
\begin{subequations}\label{e:constraints_lin_nd}
\begin{align}
\widetilde{L_c}&=\int_{\Sigma}{c e_0 \Do b+\Theta[c,e_0]},\label{eq:constraintsnondegenerate1}\\
\widetilde{J_{\mu}}&=\int_{\Sigma}{\mu\left(b\Fo+\Do\Theta{+ \frac{1}{2} \Lambda e_0e_0 b}\right)} \label{eq:constraintsnondegenerate2}
\end{align}
\end{subequations}
and the symplectic form reads
\begin{equation} \label{e:linearized_symplform}
\widetilde \varpi^\partial=\int_{\Sigma}{\delta b \delta \Theta}.
\end{equation}

The constraints \eqref{e:constraints_lin_nd}, together with the identification $\Theta= e_0 a$, are not sufficient to guarantee that  $\Do b+[a,e_0]=0$. In order to get this implication, we can exploit the freedom of the choice of $a$, given by the kernel of the map $W_{e_0}$. Indeed, as it was shown in \cite{CCS2020} for the non linearized case (a brief recap can be found in Section \ref{sec:recap_non-deg_case}), for every $\Theta$ and $b$ it is possible to find an $a$ such that 
\begin{align}\label{e:structural_constraint_lin_nd}
    e_n \Do b + e_n [a, e_0] \in \Ima W_{e_0} \quad \text{and} \quad \Theta = e_0 a
\end{align}
for a vector $e_n$ completing the set $e_0$ to a basis of $\mathcal{V}_{\Sigma}$. Then the choice \eqref{e:structural_constraint_lin_nd} together with the constraints \eqref{e:constraints_lin_nd} is equivalent to the restriction of the Euler-Lagrange equations on the bulk to the boundary.

We now compute the structure of the Poisson brackets of the constraints. We first need a technical lemma about the Hamiltonian vector fields of the constraints.
\begin{lemma}\label{lem:hamvecfields_lin_nd} 
The components of the Hamiltonian vector fields associated to the constraints of the non-degenerate linearized Palatini-Cartan theory are
\begin{subequations}
\begin{align}\label{eq:hamfJ}
    \mathbb L_b & = [c,e_0]  &
	\mathbb L_\Theta& = e_0\Do c  \\
    \mathbb J_b & = \Do\mu &
	\mathbb J_\Theta & = \mu \Fo {+ \frac{1}{2} \mu\Lambda e_0e_0 }
\end{align}
\end{subequations}
where the components of a generic vector field are defined as $\mathbb{F}=\mathbb{F}_b \frac{\delta}{\delta b} +\mathbb{F}_\Theta\frac{\delta}{\delta \Theta}$.
\end{lemma}

\begin{proof}
The Hamiltonian vector field $\mathbb{F}$ of a function $F$ satisfies
\begin{equation*}
\iota_{\mathbb F}\widetilde \varpi^{\partial}-\delta\widetilde F=0.
\end{equation*}
The result thus follow easily from the variation of the constraints:
\begin{align*}
    \delta \widetilde{L_c} &=\int_{\Sigma}{(\Do c e_0 \delta b+[c,e_0]\delta\Theta)};\\
    \delta \widetilde{J_{\mu}} &=\int_{\Sigma}{\mu\left(\Fo\delta b -\Do\delta\Theta{+ \frac{1}{2} \Lambda e_0e_0 \delta b}\right)}=\int_{\Sigma}{\left(\mu \Fo \delta b +\Do\mu \delta\Theta{+ \frac{1}{2}\mu \Lambda e_0e_0 \delta b}\right)}.
\end{align*}
\end{proof}

\begin{theorem}\label{thm:Poissonbrackets_lin_nd}
Let $g^\partial_0$ be non-degenerate. Then the Poisson algebra of constraints \eqref{e:constraints_lin_nd} is abelian and therefore the vanishing locus of such constraints defines a coisotropic submanifold. In particular
    \begin{equation}
    \{\widetilde{L_c},\widetilde{L_c}\}=0\quad 
    \{\widetilde{J_{\mu}},\widetilde{J_{\mu}}\}=0\quad
    \{\widetilde{L_c},\widetilde{J_{\mu}}\}=0.
    \end{equation}
\end{theorem}

\begin{proof}
Using the definition of the Poisson bracket of a symplectic manifold
\begin{equation*}
\{\widetilde F,\widetilde G\} = \iota_{\mathbb F}\iota_{\mathbb G}\,\widetilde\varpi^{\partial}=\iota_{\mathbb F}\delta \widetilde G,
\end{equation*}
using the results of Lemma \ref{lem:hamvecfields_lin_nd}, we get the following expression for the Poisson brackets of the constraints:
\begin{equation*}
\{\widetilde{L_c},\widetilde{L_c}\}= 2 \int_{\Sigma}{[c,e_0] e_0 \Do c}=0
\end{equation*}
since it is the integral of a total derivative given that 
\begin{align*}
    \Do ([c,c] e_0 e_0)&=\Do [c,c] e_0 e_0 =2[\Do c,c] e_0 e_0\\
&=2\Do c [c,e_0 e_0]=4[c,e_0] e_0 \Do c;
\end{align*}
\begin{equation*}
\{\widetilde{J_{\mu}},\widetilde{J_{\mu}}\}=\int_{\Sigma}{2\Do \mu \mu  \Fo {+ \Do \mu\Lambda \mu e_0e_0 }}=0,
\end{equation*}
which is equivalent to a total derivative as before, indeed
\begin{align*}
    \Do (\mu \mu  \Fo )&=\Do (\mu \mu) \Fo +\mu \mu  \Do \Fo \\
&=\Do (\mu \mu) \Fo =2\Do \mu \mu  \Fo 
\end{align*}
{and $\Do e_0=0$};
\begin{equation*}
\{\widetilde{J_{\mu}},\widetilde{L_c}\}=\int_{\Sigma}{\left(\Do \mu  e_0 \Do c+\mu  \Fo  [c,e_0]{+ \frac{1}{2} \Lambda e_0e_0 [c, e_0]}\right)}=0,
\end{equation*}
since
\begin{align*}
    \Do (\mu  e_0 \Do c )&=
    \Do \mu  e_0 \Do c+\mu  e_0 [\Fo ,c]
\end{align*}
{ and $[c, e_0^3]=0$.}
\end{proof}

This proves that the reduced phase space of the non-degenerate linearized PC theory is coisotropic. This of course also follows from the linearization of the result of Reference \cite{CS2019} on the non-degenerate Palatini-Cartan theory. 

\subsection{Degenerate boundary metric}
Let now $g^\partial_0$ be degenerate. 
In this case some of the properties useful to characterize the boundary structure of the non-degenerate case are different. In particular from Lemma \ref{lem:varrho12_deg} the map $\varrho_0|_{\mathrm{Ker}W_{e_0}^{(1,2)}}$ is no longer injective and 
\begin{align*}
   \Ima \left(\varrho_0|_{\Ker {W_{e_0}^{(1,2)}}}\right) \neq \Ker{W_{e_0}^{(2,1)}} .
\end{align*}
This implies that it is not possible to find an $a$ that solves \eqref{e:structural_constraint_lin_nd} for all $\Theta$. Digging more in the results of Section \ref{sec:technical} we get that  $\dim(\Ker{\varrho|_{\Ker{W_{e_0}^{(1,2)}}}})=2$ (Lemma \ref{lem:varrho12_deg}),
and consequently $\dim(\Ima{\varrho|_{\Ker{W_{e_0}^{(1,2)}}}})=4$.
Moreover $\dim\Ker{W_{e_0}^{(2,1)}}=6$, hence $\dim\mathcal{T}=2$. We conclude that, if we want to be able to find $a$ such that the constraints \eqref{e:constraints_lin_nd} are equivalent to the restriction of the Euler-Lagrange equations on the bulk to the boundary we have to impose two extra conditions on $\Theta$ and to modify the structural constraint \eqref{e:structural_constraint_lin_nd} accordingly.

Therefore, using Lemma \ref{lem:relationSandT}, we can add to the set of constraints the additional one
\begin{equation}
\widetilde R_\tau = \int_{\Sigma}{\tau(d_{\omega_0}b+[a,e_0])},
\end{equation}
with $\tau \in \mathcal{S}_0[1]$. Because of the definition\footnote{Here $\mathcal{S}_0$ is defined as $\mathcal{S}$ in \eqref{e:defS} but with all the maps built out of $e_0$ instead of $e$.} of $\mathcal{S}_0$ and of Lemma \ref{lem:relationSandT} we automatically have $\widetilde R_\tau[b,a+v]=\widetilde R_\tau[b,a]$ for $v\in\Ker{W_{e_0}^{(1,2)}}$. Hence the constraint $\widetilde R_\tau$ can be written in terms of $\Theta$ and $b$. Since $\int_{\Sigma}\tau [a,e_0] =- \int_{\Sigma}[\tau ,e_0] a $ we can use Lemma \ref{lem:[tau,e]inImW} and write $\widetilde R_\tau$ as
\begin{align*}
    \widetilde R_\tau 
    = \int_{\Sigma} \tau d_{\omega_0} b - W_{e_0}^{-1}([\tau ,e_0])e_0 a
    = \int_{\Sigma} \tau d_{\omega_0} b - W_{e_0}^{-1}([\tau ,e_0])\Theta .
\end{align*}

On the other hand, the structural constraint \eqref{e:structural_constraint_lin_nd} is modified as follows:
\begin{align}\label{e:structural_constraint_lin_mod}
    e_n \Do e + e_n [a, e_0]- p_{\mathcal{T}}(e_n \Do e + e_n [a, e_0]) \in \Ima W_{e_0}.
\end{align}

Collecting these results we get that the functionals defining the constraints of the linearized Palatini-Cartan theory are
\begin{align*}
\widetilde{L_c} & =\int_{\Sigma}c e_0 \Do b+\Theta[c,e_0];\\
\widetilde{J_{\mu}} & =\int_{\Sigma}\mu\left(b\Fo+\Do\Theta{+ \frac{1}{2} \Lambda e_0e_0 b}\right);\\
\widetilde R_\tau & = \int_{\Sigma}\tau d_{\omega_0}b-W_{e_0}^{-1}([\tau ,e_0])\Theta 
\end{align*}
where $\tau \in \mathcal{S}_0[1]$. 

We can now compute the Poisson brackets of the constraints to assess their type. First we need to compute the Hamiltonian vector field of the additional constraint $\widetilde{R_{\tau}}$:
\begin{lemma}\label{lem:hamvectfields_lin_deg}
The components of the Hamiltonian vector field of  $\widetilde R_\tau$ are given by
\begin{align*}
     \mathbb R_b & = W_{e_0}^{-1}([\tau ,e_0]),  &
	\mathbb R_\Theta& = \Do \tau .
\end{align*}
\end{lemma}
\begin{proof}
Trivial application of the equation $\iota_{\mathbb F}\widetilde \varpi^{\partial}-\delta\widetilde F=0.$
\end{proof}

Before proceeding to the main theorem, giving an explicit expression of the Poisson brackets of the constraints, we give some insight on Proposition \ref{prop:components_of_tau} and of Corollary \ref{cor:components_of_W-1[tau,e]} in the case $N=4$.

\begin{corollary}[of Proposition \ref{prop:components_of_tau}] \label{cor:components_of_tau_N=4}
    Let $p\in \Sigma$ and $U$ a neighbourhood of $p$, then in normal geodesic coordinates centered in $p$ and in the standard basis of $\mathcal{V}_{\Sigma}$,
the free components of an element $\tau \in \mathcal{S}_0$ are 2 and are characterized by the following equations:
\begin{align*}
    \tau_3^{abc}&=0 \quad\forall a,b,c\\
    \tau_\alpha^{123}&=0 \quad\alpha=1,2\\
    \tau_\alpha^{124}&=0 \quad\alpha=1,2\\
    \tau_2^{134} &=\frac{\tau_1^{234} {g_0}_{22}- \tau_1^{134}({g_0}_{12}+{g_0}_{21})}{{g_0}_{11}} \\
    \tau_1^{134} &=-\tau_2^{234}.
\end{align*}
Correspondingly we have
\begin{align*}
    W_{e_0}^{-1}([\tau ,e_0])_1^1 & = \tau_1^{134} & 
    W_{e_0}^{-1}([\tau ,e_0])_2^2 & = \tau_2^{234} \\
    W_{e_0}^{-1}([\tau ,e_0])_1^2 & = \tau_1^{234} & 
    W_{e_0}^{-1}([\tau ,e_0])_2^1 & = \tau_2^{134}.
\end{align*}
\end{corollary}

\begin{theorem} \label{thm:Poissonbrackets_lin_d}
    Let $g_0^\partial$ be degenerate on $\Sigma$. Then the structure of the Poisson brackets of the constraints $\widetilde{L_c}$, $\widetilde{J_{\mu}}$ and $\widetilde{R_{\tau}}$ is given by the following expressions:
    \begin{align*}
        &\{\widetilde{L_c}, \widetilde{L_c}\}  = 0 
        & \{\widetilde{L_c}, \widetilde{J_{\mu}}\}  =  0 \\
       	& \{\widetilde{J_{\mu}}, \widetilde{J_{\mu}}\}  =  0 
       	& \{\widetilde{J_{\mu}},\widetilde{R_{\tau}}\}  = \widetilde{F}_{\mu\tau} \\
       	& \{\widetilde{L_c}, \widetilde{R_{\tau}}\} =  0  
 		& \{\widetilde{R_{\tau}},\widetilde{R_{\tau}}\}  = \widetilde{F}_{\tau\tau}
    \end{align*}
    where $\widetilde{F}_{\mu\tau}$ and $\widetilde{F}_{\tau\tau}$ are functions of the background fields $e_0$, $\omega_0$ and of $\mu$ and $\tau$. These functions vanish if $\tau$ is covariantly constant.
\end{theorem}
\begin{proof}
The brackets between the constraints $\widetilde{L_c}$ and  $\widetilde{J_{\mu}}$ are the same as in the non-degenerate case and have already been computed in Theorem \ref{thm:Poissonbrackets_lin_nd}. Let us now compute $\{\widetilde{L_c}, \widetilde{R_{\tau}}\}$. Using the results of Lemmas \ref{lem:hamvecfields_lin_nd} and \ref{lem:hamvectfields_lin_deg} we get  
\begin{align*}
    \{\widetilde{L_c}, \widetilde{R_{\tau}}\} & = \int_{\Sigma} [c,e_0] \Do \tau + \Do c [\tau ,e_0] = \int_{\Sigma}\Do ([c,e_0] \tau ) = 0
\end{align*}
where we used that $\Do e_0 =0$ and that $\Sigma$ is closed. Consider now $\{\widetilde{J_{\mu}},\widetilde{R_{\tau}}\}$:
\begin{align*}
    \{\widetilde{J_{\mu}},\widetilde{R_{\tau}}\} & = \int_{\Sigma} \Do\mu \Do \tau + \mu \Fo W_{e_0}^{-1}([\tau ,e_0]){+ \frac{1}{2} \mu \Lambda e_0[\tau ,e_0]}.
\end{align*}
{ We first note that the last term vanishes. Then, 
since the remaining terms do not depend on $b$, the bracket} cannot be proportional to any of the constraints. We now want  to prove, using coordinates, that this bracket does not vanish. Integrating by parts the first term and discarding the total derivative ($\Sigma$ closed) we get 
\begin{align*}
    \int_{\Sigma} \Do\mu \Do \tau = \int_{\Sigma} - \mu [ \Fo , \tau ] = \int_{\Sigma}  \tau [ \Fo , \mu ]. 
\end{align*}
We now split the computation in two parallel ways, one for the components of $\mu$ proportional to the image of $e_0$ (on the boundary) and the other for the orthogonal part of $\mu$. We parametrize $\mu$ with a vector field $\xi \in \Gamma (T\Sigma)$, such that $\mu=\iota_{\xi}e_0 + \mu^4 e_n$. Let us denote by $F$ a function of the last component $\mu^4 e_n$. We have:
\begin{align*}
    \{\widetilde{J_{\mu}},\widetilde{R_{\tau}}\} & = \int_{\Sigma}  \tau [ \Fo , \iota_{\xi} e_0 ] + \iota_{\xi} e_0 \Fo W_{e_0}^{-1}([\tau ,e_0])+ F(\mu^4 e_n)\\
    & = \int_{\Sigma}  \tau[\Fo, \iota_{\xi} e_0]+ \iota_{\xi} \Fo [\tau ,e_0]+ F(\mu^4 e_n)\\
    & = \int_{\Sigma} \tau \iota_{\xi}[\Fo, e_0]+ F(\mu^4 e_n) = F(\mu^4 e_n)
\end{align*}
where we used that the couple $(e_0, \omega_0)$ solves the equations $e_0 \Fo=0$ and $\Do e_0=0$. Consequently the last term vanishes since $[\Fo,e_0]= \Do (\Do e_0)= \Do(0)=0 $.

 In order to compute $F(\mu^4 e_n)$ we will make use of Corollary \ref{cor:components_of_tau_N=4} and since there are no derivatives involved here, we can also simplify the result by taking $g^{\partial}_0$ diagonal and working in the point $p$ (basis point of the normal geodesic coordinates). Furthermore this approach is also suitable for proving the same result of the tangent part, but it is way more complicated, involving the computation of all the components of the quantities appearing in the bracket. Hence, in the standard basis we have
\begin{align*}
    \tau [ \Fo , \mu ] &  = \left(\tau_2^{134} [ \Fo , \mu ]_{13}^{2} + \tau_1^{234} [ \Fo , \mu ]_{23}^{1} - \tau_1^{134} [ \Fo , \mu ]_{23}^{2} - \tau_2^{234} [ \Fo , \mu ]_{13}^{1}\right) \mathbf{V}\\
    &  = \left(\tau_1^{234} ([ \Fo , \mu ]_{23}^{1}+ [ \Fo , \mu ]_{13}^{2}  )- \tau_1^{134} ([ \Fo , \mu ]_{23}^{2} - [ \Fo , \mu ]_{13}^{1})\right) \mathbf{V},
\end{align*}
\begin{align*}
    \mu \Fo W_{e_0}^{-1}([\tau ,e_0]) & = \left(-(\mu \Fo)_{23}^{234} \tau_{1}^{134}-(\mu \Fo)_{13}^{134} \tau_{2}^{234}\right) \mathbf{V} \\
    & \quad + \left((\mu \Fo)_{13}^{234} \tau_{2}^{134}+ (\mu \Fo)_{23}^{134} \tau_{1}^{234}\right) \mathbf{V} \\
    & = \left(\tau_{1}^{134}( (\mu \Fo)_{23}^{234} -(\mu \Fo)_{13}^{134})\right) \mathbf{V} \\
    & \quad- \left(\tau_{1}^{234}((\mu \Fo)_{23}^{134} +(\mu \Fo)_{13}^{234})\right) \mathbf{V} 
\end{align*}
where $\mathbf{V}= {e_0}_1 {e_0}_2 {e_0}_3 e_n dx^1 dx^2 dx^3$. Using the two identities
\begin{align*}
    [ \Fo , \mu ]_{ab}^{i}= \sum_{k,j}\Fo_{ab}^{ij}\eta_{jk} \mu^k \; \text{and} \; (\mu \Fo)_{ab}^{ijk}= \mu^i \Fo_{ab}^{jk} + \mu^j \Fo_{ab}^{ki} + \mu^k \Fo_{ab}^{ij}
\end{align*}
we can further simplify $\tau [ \Fo , \mu ] + \mu \Fo W_{e_0}^{-1}([\tau ,e_0])$. A simple computation shows (as it should be because of the first part of this proof) that the terms containing $\mu^3$ are the same with opposite sign in the two summands, hence they cancel and the terms containing $\mu^1$ and $\mu^2$ vanish because they contain components of $\Fo$ that are zero due to the equations $e_0 \Fo=0$ and $\Do e_0=0$. Hence we are left with the terms containing $\mu^4 e_n$
and get 
\begin{align*}
   F(\mu^4 e_n) & = \left(\tau_1^{234} (\Fo_{23}^{13} \mu^4 + \Fo_{13}^{23} \mu^4) - \tau_1^{134}(\Fo_{23}^{23}\mu^4 -\Fo_{13}^{13}\mu^4)\right) \mathbf{V}\\
    & + \left(\tau_1^{234} (\mu^4 \Fo_{23}^{13}+ \mu^4 \Fo_{13}^{23}) - \tau_1^{134}(\mu^4 \Fo_{23}^{23}-\mu^4 \Fo_{13}^{13})\right) \mathbf{V}\\
  & = 2 \mu^4\left(-\tau_1^{234} (\Fo_{23}^{13}  + \Fo_{13}^{23} ) + \tau_1^{134}(\Fo_{23}^{23} -\Fo_{13}^{13})\right) \mathbf{V}.
\end{align*}
These expression does not vanish since none of the equations $e_0 \Fo=0$ and $\Do e_0=0$ contain these components of $\Fo$. We denote by $\widetilde{F}_{\mu\tau}= F(\mu^4 e_n)$ and note that if $\Do \tau=0$ (i.e. $\tau$ covariantly constant) it vanishes.
Let us now pass to the last bracket:
\begin{align*}
    \{\widetilde{R_{\tau}},\widetilde{R_{\tau}}\} & = \int_{\Sigma} W_{e_0}^{-1}([\tau ,e_0])\Do \tau.
\end{align*}
As in the previous case this bracket does not depend on $b$ and $\Theta$, hence it cannot be proportional to any constraint. Again, we want to prove that this expression does not vanish and in order to reach this goal we will work in coordinates using the results of Corollary \ref{cor:components_of_tau_N=4}. Note that since the expression of the bracket entails taking derivatives, we need to work in a neighbourhood and retain the complete expression for the relations between the components of $\tau$.
\begin{align}
    \{\widetilde{R_{\tau}},\widetilde{R_{\tau}}\} & =\int_{\Sigma} \left( -\tau_1^{134} \partial_3 \tau_2^{234}-\tau_2^{234} \partial_3 \tau_1^{134}+ \tau_2^{134} \partial_3 \tau_1^{234}+ \tau_1^{234} \partial_3 \tau_2^{134}\right) \mathbf{V} \nonumber\\
    & = \int_{\Sigma}\left( 2\tau_1^{134} \partial_3 \tau_1^{134}+ \tau_1^{234} \partial_3 \tau_1^{234}+ \tau_1^{234} \partial_3 \frac{\tau_1^{234} {g_0}_{22}-2 \tau_1^{134}{g_0}_{12}}{{g_0}_{11}}\right) \mathbf{V}. \label{e:RR_lin_4}
\end{align}
For a generic (not covariantly constant) $\tau$ and  a generic background metric $g_0$ this expression does not vanish. We denote this last quantity with $\widetilde{F}_{\tau\tau}$.
\end{proof}

\begin{corollary}
    If $\tau$ is not covariantly constant, the constraints $\widetilde{L_c}$, $\widetilde{J_{\mu}}$ and $\widetilde{R_{\tau}}$ do not form a first class system. In particular $\widetilde{R_{\tau}}$ is a second class constraint while the others are first class.
\end{corollary}
\begin{proof}
We prove this result using Proposition \ref{prop:numberofsecondclassconstraints} and subsequently by trivially applying Lemma \ref{lem:vanishofBD-1BT}.
\end{proof}

{
\begin{remark}\label{rmk:zeromodes_linearized}
In case $\tau$ is covariantly constant, i.e., $d_{\omega_0} \tau=0$, also the constraint $\widetilde{R_{\tau}}$ is first class. Indeed, we get
\begin{align*}
    \{\widetilde{R_{\tau}},\widetilde{R_{\tau}}\} & = \int_{\Sigma} W_{e_0}^{-1}([\tau ,e_0])\Do \tau=0.
\end{align*}
These are usually called first-class zero mode \cite{AlexandrovSpeziale15}. The interpretation of these zero modes in terms of symmetries is still unknown and will be the object of future studies.
\end{remark}

These results can be compared with the ones in \cite{ECT1987} where the same problem is analysed in the Einstein--Hilbert formalism. In this article the authors found $2N$ first class constraints (shared with the time/space-like case) and additional $N(N-3)/2$ second class constraints. Despite the different number of first-class constraints, the results here outlined coincide exactly with those that we have found. Indeed we found exactly the same amount of additional second-class constraints, while the difference in the number of first-class constraint is due to the different formalism adopted. Indeed, in the space- or time-like case in the Einstein--Hilbert formalism one has $2N$ first-class constraints \cite{CS2016a}, while in the Palatini--Cartan formalism one has $N(N+1)/2$ first-class constraints \cite{CCS2020}. This difference is due to the different number of degrees of freedom of the space of boundary fields. Note however that the numbers of physical degrees of freedom are the same. The actual comparison of the results of the present article and those in \cite{ECT1987} can be made more precise through a procedure similar to that outlined in \cite[Theorem 4.25]{CS2019} and will be the object of future work.

}

\section{Long Proofs} \label{sec:appendix_long_proofs}

All the proofs of Lemmas \ref{lem:We_bulk} and \ref{lem:We_boundary} can be found in \cite[Appendix A]{CCS2020}. We recall here verbatim the partial proof of the point \ref{lem:kernel12-21} of Lemma \ref{lem:We_boundary} in order to have a useful reference for the new proofs.
\begin{proof}[Proof of Lemma
{\hyperref[lem:kernel12-21]{\ref*{lem:We_boundary}.(\ref*{lem:kernel12-21})}} ] (\cite{CCS2020})
Consider $W_{N-3}^{ \partial, (1,2)}\colon \Omega_{N, \partial}^{1,2}  \longrightarrow \Omega_{N, \partial}^{N-2,N-1}$.

The dimensions of domain and codomain of this map are respectively 
$\dim \Omega_{N, \partial}^{1,2} = (N-1)\frac{N(N-1)}{2}$ and $\dim \Omega_{N, \partial}^{N-2,N-1} = (N-1)N$. The kernel of $W_{N-3}^{ \partial, (1,2)}$ is defined by the following set of equations:
\begin{align*}
X_{\mu_1}^{ab} e_a e_b \wedge e_{\mu_2} \wedge \dots \wedge e_{\mu_{N-2}} dx^{\mu_1}dx^{\mu_2}\dots dx^{\mu_{N-2}}=0
\end{align*}
where we used $e_a$ as a basis for $\mathcal{V}_{\Sigma}$. Let now $k=N$ be the transversal direction and let $k'\in\{1, \dots N-1\}$. Since $\{dx^{\mu_1}dx^{\mu_2}\dots dx^{\mu_{N-2}}\}$ is a basis for $\Omega_{\partial}^{N-2}$ we obtain $N-1$ equations of the form
\begin{align*}
\sum_{\sigma} X_{\mu_{\sigma(1)}}^{ab} e_a e_b \wedge e_{\mu_{\sigma(2)}} \wedge \dots \wedge e_{\mu_{\sigma(N-2)}} =0  
\end{align*}
where $\sigma$ runs on all permutations of $N-2$ elements and $1\leq \mu_i \leq N-1$, $\mu_i \neq k'$ for all $1\leq i \leq N-2$. Recall now that $e_a e_b \wedge e_{\mu_{\sigma(2)}} \wedge \dots \wedge e_{\mu_{\sigma(N-2)}}$ is a basis of $\wedge^{N-1}\mathcal{V}_{\Sigma}$. Hence we obtain the following equations:
\begin{align*}
 X_{i}^{N k'} =0 \qquad & 1 \leq i \leq N-1 \; i \neq k', \\
 \sum_{i\neq N, i\neq k'} X_{i}^{iN} =0, \qquad & 
  \sum_{i\neq N, i\neq k'} X_{i}^{ik'} =0 .
\end{align*}
Letting now $k'$ vary in $\{1, \dots , N-1\}$ we obtain the following equations:
\begin{subequations}\label{e:conditions_kernel_v}
\begin{align}
 X_{i}^{N j} =0 \qquad & 1 \leq i,j \leq N-1 \; i \neq j ,\\
 \sum_{i\neq N, i\neq j} X_{i}^{iN} =0 \qquad & 1 \leq j \leq N-1, \\
  \sum_{i\neq N, i\neq j} X_{i}^{ij} =0 \qquad & 1 \leq j \leq N-1 .
\end{align}
\end{subequations}
It is easy to check that these equations are independent.
The total number of equations defining the kernel is then $(N-1)+(N-1)(N-2)+(N-1) = (N-1)N$ which coincides with number of degrees of freedom of the codomain. Hence  $W_{N-3}^{ \partial, (1,2)}$ is surjective but not injective. In particular $\dim \text{Ker} W_{N-3}^{\partial, (1,2)}= (N-1)\frac{N(N-1)}{2}- N(N-1)= \frac{N(N-1)}{2}(N-3)$.
\end{proof}

\begin{proof}[Proof of Lemma \ref{lem:varrho12_deg}]
Consider $\varrho |_{\text{Ker} W_{N-3}^{\partial, (1,2)}}: \text{Ker} W_{N-3}^{\partial, (1,2)} \rightarrow \Omega_{N, \partial}^{2,1}$. From the proof of Lemma  \hyperref[lem:kernel12-21]{\ref*{lem:We_boundary}.(\ref*{lem:kernel12-21})} we know that $\dim \text{Ker} W_{N-3}^{\partial, (2,1)}= \frac{N(N-1)}{2}(N-3)$. An element $v \in  \text{Ker} W_{N-3}^{\partial, (1,2)}$ must satisfy equations 
the following equations:
\begin{subequations}\label{e:eq_for_comp_of_v}
\begin{align} 
v_{i}^{N j} =0 \qquad & 1 \leq i,j \leq N-1 \; i \neq j, \\
\sum_{i\neq N, i\neq j} v_{i}^{iN} =0 \qquad & 1 \leq j \leq N-1, \\
 \sum_{i\neq N, i\neq j} v_{i}^{ij} =0 \qquad & 1 \leq j \leq N-1. \label{e:constraint_kerW12_n3}
\end{align}
\end{subequations}

The kernel of $\varrho$ is defined by the following set of equations\footnote{ \label{foot:brackets} Here we use that in every point we can find a basis in $\mathcal{V}_{\Sigma}$ such that $e_\mu^i= \delta_\mu^i$: $[v,e]_{\mu_1\mu_2}^a= v_{\mu_1}^{ab}\eta_{bc}e_{\mu_2}^c- v_{\mu_2}^{ab}\eta_{bc}e_{\mu_1}^c= v_{\mu_1}^{ab}e_{b}^d\eta_{dc}e_{\mu_2}^c- v_{\mu_2}^{ab}e_{b}^d\eta_{dc}e_{\mu_1}^c$}:
\begin{align*}
[v,e]_{\mu_1\mu_2}^a= v_{\mu_1}^{ab} g^{\partial}_{b \mu_2} - v_{\mu_2}^{ab}g^{\partial}_{b\mu_1}=0.
\end{align*}
Using now normal geodesic coordinates, we can diagonalise $g^\partial$ with eigenvalues on the diagonal $\alpha_{\mu} \in \{1,-1,0\}$:
\begin{align} \label{e:rho_normalgeodesiccoordinates}
[v,e]_{\mu_1\mu_2}^a= v_{\mu_1}^{a\mu_2} \alpha_{ \mu_2} - v_{\mu_2}^{a\mu_1} \alpha_{\mu_1}=0.
\end{align}
Let now $\alpha_{ \mu}=0$ for $  \mu = N-1$ and $\alpha_{ \mu}= \pm 1 $ for $ 1 \leq \mu \leq N-2$.
We adopt now the following convention on indices for $m,p\in \mathbb{N}$: $ 1 \leq i_m \leq N-2$, $i_m \neq i_p $ iff $m\neq p$.
Using $v \in  \text{Ker} W_{N-3}^{\partial, (2,1)}$, from \eqref{e:rho_normalgeodesiccoordinates} we get
\begin{subequations}
\begin{align}
    [v,e]_{i_1 i_2}^{i_3} & = v_{i_1}^{ i_3 i_2} - v_{i_2}^{ i_3 i_1} & [v,e]_{i_1 i_2}^{i_1} & = v_{i_1}^{i_1 i_2} \label{e:kernel_rho_n1} \\
    [v,e]_{i_1 i_2}^{N-1} & = v_{i_1}^{N-1 i_2} - v_{i_2}^{N-1 i_1} & [v,e]_{i_1 i_2}^{N} & = 0 \\
    [v,e]_{i_1 N-1}^{i_2} & = - v_{N-1}^{i_2 i_1} & [v,e]_{i N-1}^{N-1} & = - v_{N-1}^{N-1 i}  \label{e:kernel_rho_n5} \\
    [v,e]_{i N-1}^{i} & =0 & [v,e]_{i N-1}^{N} & = 0.
\end{align}
\end{subequations}
By imposing that every component vanishes, we get the corresponding equations for the kernel.
It is easy to check that these equations are independent but the second of \eqref{e:kernel_rho_n1} and the first of \eqref{e:kernel_rho_n5} which are connected by \eqref{e:constraint_kerW12_n3}.
The total number of equations defining the kernel is then 
\begin{align*}
    \frac{(N-2)(N-3)(N-4)}{2} +(N-3)(N-4)+  (N-2)(N-3).
\end{align*}
Since  $\frac{N(N-1)}{2}(N-3)$ is the number of degrees of freedom of the domain, the kernel has dimension \begin{align*}
    \dim (\text{Ker}\varrho |_{\mathrm{Ker} W_{N-3}^{\partial, (1,2)}})= \frac{N(N-3)}{2}.
\end{align*}
\end{proof}

\begin{proof}[Proof of Proposition \ref{prop:components_of_tau}]
We split the proof into simpler lemmas. First we compute the dimension of $\mathcal{S}$ and the equations defining it at the point $p$ in Lemmas \ref{lem:deftau_W} and \ref{lem:deftau_rho}. 

Collecting these results we get that $\mathcal{S}$ is defined by the following equations:
    \begin{align*}
        \sum_{\mu_1 \dots \mu_{N-3}=1}^{N-2} X_{\mu_1 \dots \mu_{N-3}}^{N N-1 \mu_1 \dots \mu_{N-3}}=0 & \\
        X_{N-1 i_{1}\dots i_{n-4}}^{\mu_1 \dots \mu_{N-1}}=0 & 
            \quad 1 \leq \mu_a \leq N, \; 1 \leq i_a \leq N-2 \\
        X_{ i_1 \dots i_{N-3}}^{N i_1 \dots i_{N-2}}=0 &
            \quad  1 \leq i_a \leq N-2 \\
        X_{ i_1 \dots i_{N-3}}^{N-1 i_1 \dots i_{N-2} }=0 &
            \quad  1 \leq i_a \leq N-2 \\
        X_{i_1\dots i_{N-4} i_{N-3}}^{N N-1 i_1\dots i_{N-4}  i_{N-2}}+X_{i_1\dots i_{N-4} i_{N-2}}^{N N-1 i_1\dots i_{N-4}  i_{N-3}}=0 &    \quad 1 \leq i_a \leq N-2.
    \end{align*}
Counting them and subtracting the result to the total dimension of $\Omega^{N-3,N-1}_{\partial}$ we get the claimed result.
Then in Lemma \ref{lem:tau_rho_neigh}  we express the equations defining the kernel of $\widetilde{\varrho}^{(N-3, N-1)}$ in the geodesic neighbourhood $U$ in terms of the components of $\tau \in \mathcal{S}$ and those of the modified metric $\widetilde{g}^{\partial} \coloneqq \eta (\widetilde{e},\widetilde{e}) $ and find the corresponding free components. Note also that the equations in Lemma \ref{lem:deftau_W} hold in a neighbourhood since we are not using normal geodesic coordinates in the proof.
\end{proof}

\begin{lemma}  \label{lem:deftau_W}
The space $\mathrm{Ker}W_{1}^{\partial (N-3,N-1)} \subset \Omega^{N-3,N-1}_{\partial}$ in the standard basis is defined by the following $N-1$ equations
\begin{align*}
    \sum_{\substack{\mu_1 \dots \mu_{N-3}=1\\ \mu_i \neq k, \mu_i \neq \mu_j}}^{N-1} X_{\mu_1 \dots \mu_{N-3}}^{N k \mu_1 \dots \mu_{N-3}} =0 \qquad 1 \leq k \leq N-1.
\end{align*}

\end{lemma}

\begin{proof}
 Consider $W_{1}^{ \partial, (N-3,N-1)}: \Omega_{\partial}^{N-3,N-1}  \longrightarrow \Omega_{\partial}^{N-2,N}$:
the kernel of it is defined by the following set of equations:
\begin{align*}
X_{\mu_1 \dots \mu_{N-3}}^{a \dots a_{N-1}} e_a \wedge \dots \wedge e_{a_{N-1}} \wedge e_{\mu_{N-2}} dx^{\mu_1}dx^{\mu_2}\dots dx^{\mu_{N-2}}=0
\end{align*}
 where we used $e_a$ as a basis for $\mathcal{V}_{\Sigma}$. 
 Since $\{dx^{\mu_1}dx^{\mu_2}\dots dx^{\mu_{N-2}}\}$ is a basis for $\Omega_{\partial}^{N-2}$ we obtain $N-1$ equations of the form
 \begin{align*}
\sum_{\sigma} X_{\mu_{\sigma(1)} \dots \mu_{\sigma(N-3)}}^{a \dots a_{N-1}} e_a \wedge \dots \wedge e_{a_{N-1}} \wedge e_{\mu_{\sigma(N-2)}} =0
\end{align*}
where $\sigma$ runs on all permutations of $N-2$ elements and $1\leq \mu_i \leq N-1$ and denote by $k$ the missing index: $\mu_i \neq k$ for all $1\leq i \leq N-2$. 
Recall now that $e_a \wedge e_{\mu_{\sigma(2)}} \wedge \dots \wedge e_{\mu_{\sigma(N-2)}}$ is a basis of $\wedge^{N} \mathcal{V}_{\Sigma}$. Hence we obtain the following $N-1$ equations:
\begin{align*}
    \sum_{\substack{\mu_1 \dots \mu_{N-3}=1\\ \mu_i \neq k, \mu_i \neq \mu_j}}^{N-1} X_{\mu_1 \dots \mu_{N-3}}^{N k \mu_1 \dots \mu_{N-3}} =0 \qquad 1 \leq k \leq N-1.
\end{align*}
\end{proof}

\begin{lemma}\label{lem:deftau_rho}
    The space $\mathrm{Ker} \tilde{\rho} |_{\Omega^{N-3,N-1}_{\partial}} \subset \Omega^{N-3,N-1}_{\partial}$ is defined by the following four sets of equations in the standard basis:
    \begin{align*}
        X_{N-1  i_{1}\dots i_{n-4}}^{\mu_1 \dots \mu_{N-1}}=0 & 
            \quad 1 \leq \mu_a \leq N, \; 1 \leq i_a \leq N-2 \\
        X_{ i_1 \dots i_{N-3}}^{N i_1 \dots i_{N-2} }=0 &
            \quad  1 \leq i_a \leq N-2 \\
        X_{i_1 \dots i_{N-3}}^{ i_1 \dots i_{N-2} N-1}=0 &
            \quad 1 \leq i_a \leq N-2 \\
        X_{ i_1\dots i_{N-4} i_{N-3}}^{N N-1 i_1\dots i_{N-4} i_{N-2}} + & X_{ i_1\dots i_{N-4} i_{N-2}}^{N N-1 i_1\dots i_{N-4} i_{N-3}} =0  \quad  1 \leq i_a \leq N-2.
    \end{align*}
\end{lemma}

\begin{proof}
    
	In normal geodesic coordinates the boundary metric $g^\partial$ at the base point is diagonal, and we can assume that its eigenvalues $\alpha_i$ are such that $\alpha_a= 1$ for $1\leq a \leq N-2$ and $\alpha_{N-1}=0$. Since $X$ is such that  $\iota_X g^\partial=0$ we get $X = \partial_{N-1}$. 
	Let now be $\beta= \sum_{i=1}^{N-1} \beta_i dx^i $ a generic one form. From the equation $\iota_X \beta=1 $ we get $\beta_{N-1}=1$.  Hence $\widehat{e}= \sum_{i=1}^{N-1} \beta_i dx^i e_{N-1}$ and $\widetilde{e}= \sum_{i=1}^{N-2}(e_{i}   + \beta_i e_{N-1})dx^{i}$. 
	We now impose the last condition to find an explicit expression for $\widetilde{e}$ in the standard basis.
	
	Using these coordinates, since $X = \partial_{N-1}$ we can take $Y^i_0$ as $Y^i_0=\partial_i$.
	Let now $v \in \Omega_{\partial}^{1,2}$ such that $e^{N-3}v=0$ i.e. its components must satisfy \eqref{e:eq_for_comp_of_v}. Using the same techniques as in the proof of Lemma \hyperref[lem:kernel12-21]{\ref*{lem:We_boundary}.(\ref*{lem:kernel12-21})}, we get 
	\begin{align*}
        v e^{N-4} \widetilde{e} & = X_{\mu_1}^{ab} e_a e_b \wedge e_{\mu_2} \wedge \dots \wedge e_{\mu_{N-3}} \widetilde{e}_{\mu_{N-2}}^c e_c dx^{\mu_1}dx^{\mu_2}\dots dx^{\mu_{N-2}}\\
        & = X_{\mu_1}^{ab} e_a e_b \wedge e_{\mu_2} \wedge \dots \wedge e_{\mu_{N-3}} e_{\mu_{N-2}} dx^{\mu_1}dx^{\mu_2}\dots dx^{\mu_{N-2}}\\
        & \quad - X_{\mu_1}^{ab} e_a e_b \wedge e_{\mu_2} \wedge \dots \wedge e_{\mu_{N-3}} \beta_{\mu_{N-2}}e_{N-1} dx^{\mu_1}dx^{\mu_2}\dots dx^{\mu_{N-2}}
    \end{align*}
	where in the second and third line $\mu_{N-2}$ cannot take the value $N-1$. Equating this quantity to zero, we get the following equations for the components: 
	\begin{align*}
        \sum_{\sigma} X_{\mu_{\sigma(1)}}^{ab} e_a e_b  e_{\mu_{\sigma(2)}}  \dots  e_{\mu_{\sigma(N-2)}}- X_{\mu_{\sigma(1)}}^{ab} e_a e_b  e_{\mu_{\sigma(2)}}  \dots  e_{\mu_{\sigma(N-3)}} e_{N-1} \beta_{\mu_{\sigma(N-2)}}=0
    \end{align*}
    where $\mu_{\sigma(N-2)} \neq N-1$. Now, letting
    $\{\mu_1 \dots \mu_{N-2}\}=\{1 \dots N-2\}$ we get
    \begin{align*}
         & \sum_{i\neq j,N-1,N} \left( X_{j}^{N-1 N} - X_{j}^{iN} \beta_i + X_{i}^{iN} \beta_j\right)=0 \qquad j = 1 \dots N-2\\
         & \sum_{i\neq N-1,N} X_{i}^{i N} =0 \\
         & \sum_{i\neq N-1,N} X_{i}^{i N-1}- \sum_{i,j\neq N-1,N} X_{i}^{i j}\beta_j =0.
    \end{align*}
    
    Using the properties \eqref{e:eq_for_comp_of_v}, we can deduce from the very first equation that $\beta_i=0$ for $i=1 \dots N-2$. Plugging this result into the others we do not get any further condition, as all the quantities vanish automatically. 
	We deduce that, with this choice of the coordinates $x_i$, $\widetilde{e}= \sum_{i=1}^{N-2}e_{i}dx^i$.

	Now, using the same procedure as in Lemma \ref{lem:varrho12_deg} we obtain the following equations defining the kernel of $\tilde{\rho}$:
    \begin{align*}
        [\tau, e]_{\mu_1 \dots \mu_{N-2}}^{\nu_1 \dots \nu_{N-2}}= \sum_{\sigma_{N-2}} \tau_{\mu_{\sigma(1)}\dots \mu_{\sigma(N-3)}}^{\nu_1 \dots \nu_{N-2}\mu_{\sigma(N-2)}} \alpha_{\mu_{\sigma(N-2)}} = 0 
    \end{align*}
    where  $ 1\leq \mu_a \leq N-1$, $ 1\leq \nu_a \leq N$, $\alpha_a=1$ for $1\leq a \leq N-2$, $\alpha_{N-1}=0$ and $ \sigma_{N-2}$ represents the permutation of $N-2$ elements. Using the properties of the $\alpha$s we get
    \begin{subequations}\label{e:varrho_Nk_kernel}
    \begin{align}
        \sum_{\sigma_{N-3}} \tau_{N-1 i_{\sigma(1)}\dots i_{\sigma(N-4)}}^{\nu_1 \dots \nu_{N-2}i_{\sigma(N-3)}} = 0 &
            \quad  1 \leq i_a \leq N-2 \\
        \sum_{\sigma_{N-2}} \tau_{i_{\sigma(1)}\dots i_{\sigma(N-3)}}^{\nu_1 \dots \nu_{N-2}i_{\sigma(N-2)}} = 0 &
            \quad 1 \leq i_a \leq N-2
    \end{align}
    \end{subequations}
 
    for $1 \leq \nu_a \leq N$. Let us consider the first set of equations. If $\{\nu_1, \dots ,\nu_{N-2} \} \supset \{i_1, \dots ,i_{N-3} \} $ no term survives and we do not get equations. Let now $n$ be an index in $\{i_1, \dots ,i_{N-3} \}$ but not in $\{\nu_1, \dots ,\nu_{N-2} \}$: then only one term survives and we have the following equations:
    \begin{align*}
        \tau_{N-1 i_{1}\dots i_{N-4}}^{\nu_1 \dots \nu_{N-2}n} = 0 
    \end{align*}
    where $1 \leq i_a, n \leq N-2$ and  $\{\nu_1, \dots ,\nu_{N-2} \} \supset \{i_1, \dots ,i_{N-4} \} $. The only other case that is left is when there are two indices $n_1, n_2$ in $\{i_1, \dots ,i_{N-3} \}$ but not in $\{\nu_1, \dots ,\nu_{N-2} \}$: here two terms of the sum are surviving and we get:
    \begin{align*}
        \tau_{N-1  i_{1}\dots i_{N-5}n_1}^{N N-1 i_{1}\dots i_{N-5} n_3 n_2} +
        \tau_{N-1  i_{1}\dots i_{N-5}n_2}^{N N-1 i_{1}\dots i_{N-5} n_3 n_1} = 0 
    \end{align*}
    where $1 \leq i_a, n \leq N-2$ and $n_3$ is the only index left different from all the others. Because of the arbitrariness of $n_1, n_2, n_3$, this set of equations will contain also the ones corresponding to permutations of them:  
    \begin{align*}
       & \tau_{N-1  i_{1}\dots i_{N-5}n_1}^{N N-1 i_{1}\dots i_{N-5} n_2 n_3} +
        \tau_{N-1  i_{1}\dots i_{N-5}n_3}^{N N-1  i_{1}\dots i_{N-5} n_2 n_1} = 0 \\
         & \tau_{N-1  i_{1}\dots i_{N-5}n_2}^{N N-1  i_{1}\dots i_{N-5} n_1 n_3} +
        \tau_{N-1  i_{1}\dots i_{N-5}n_3}^{N N-1  i_{1}\dots i_{N-5} n_1 n_2} = 0 .
    \end{align*}
    Composing these three equations we get that
    \begin{align*}
        \tau_{N-1 i_{1}\dots i_{N-5}n_1}^{N N-1  i_{1}\dots i_{N-5} n_3 n_2}=0.
    \end{align*}
    Together with the first case this proves the first set of equations in the statement.
    We proceed in the same way for the second set in \eqref{e:varrho_Nk_kernel}. If $\{\nu_1, \dots ,\nu_{N-2} \} \supset \{i_1, \dots ,i_{N-3} \} $ no term survives and we do not get equations. Let now $n$ be an index in $\{i_1, \dots ,i_{N-3} \}$ but not in $\{\nu_1, \dots ,\nu_{N-2} \}$. We get
    \begin{align*}
        X_{ i_1 \dots i_{N-3}}^{N i_1 \dots i_{N-3}  n}=0 &
            \quad  1 \leq i_a \leq N-2 \\
        X_{ i_1 \dots i_{N-3}}^{N-1 i_1 \dots i_{N-3}  n}=0 &
            \quad  1 \leq i_a \leq N-2    
    \end{align*}
    which are respectively the second and the third set of equations in the statement. When there are two indices $n_1, n_2$ in $\{i_1, \dots ,i_{N-3} \}$ but not in $\{\nu_1, \dots ,\nu_{N-2} \}$ we get the fourth set of equations:
    \begin{align*}
        X_{ i_1\dots i_{N-4} n_1}^{N N-1 i_1\dots i_{N-4} n_2} + & X_{ i_1\dots i_{N-4} n_2}^{N N-1 i_1\dots i_{N-4} n_1} =0  \quad    \quad  1 \leq i_a \leq N-2.
    \end{align*}

\end{proof}

\begin{lemma}\label{lem:tau_rho_neigh} 
    Let $p \in \Sigma$ and $U$ an open neighbourhood of $p$. Then, in the standard basis of $\mathcal{V}_{\Sigma}$, the equations defining the space $\mathrm{Ker} \tilde{\rho} |_{\Omega^{N-3,N-1}_{\partial}} \subset \Omega^{N-3,N-1}_{\partial}$ are 
    \begin{align*}
        X_{N-1 i_{1}\dots i_{n-4}}^{\mu_1 \dots \mu_{N-1}}=0 & 
            \quad 1 \leq \mu_a \leq N, \; 1 \leq i_a \leq N-2 \\
        X_{ i_1 \dots i_{N-3}}^{N i_1 \dots i_{N-2}}=0 &
            \quad  1 \leq i_a \leq N-2 \\
        X_{ i_1 \dots i_{N-3}}^{N-1 i_1 \dots i_{N-2} }=0 &
            \quad  1 \leq i_a \leq N-2 \\
        X_{i_1\dots i_{N-4} i_{N-3}}^{N N-1 i_1\dots i_{N-4}  i_{N-2}} = & f(\widetilde{g}^{\partial}, X_{i_1\dots i_{N-4} i_{N-2}}^{N i_1\dots i_{N-4} N-1 i_{N-3}}, X_{i_1\dots i_{N-3}}^{N N-1 i_1\dots i_{N-3}}) 
    \end{align*}
    for some function $f$.
\end{lemma}
\begin{proof}
    Using the standard basis of $\mathcal{V}_{\Sigma}$, we obtain the following equations for the kernel of $\tilde{\rho}$:
    \begin{align}\label{e:tilderho_neigh}
        [\tau, \tilde{e}]_{j_1 \dots j_{N-2}}^{i_1 \dots i_{N-2}}= \sum_{\sigma, \mu} \tau_{j_{\sigma(1)}\dots j_{\sigma(N-3)}}^{i_1 \dots i_{N-2}\mu} \widetilde{g}^{\partial}_{\mu j_{\sigma(N-2)}}=0
    \end{align}
    where $\sigma$ runs over the permutations of order $N-2$ and $\mu=1 \dots N-2$, $i_k \in \{1\dots N-1\}$, $j_k \in \{1\dots N\}$. Using normal geodesic coordinates, $\widetilde{g}^{\partial}$ is diagonal in the point $p$, with diagonal entries different from zero. Hence using continuity, in the whole neighbourhood $U$ (eventually shrinking it if necessary) the diagonal component will be non-zero. Furthermore, $\det\widetilde{g}^{\partial}\neq 0$, since  $\widetilde{g}^{\partial}$ is non-degenerate by construction. \\
    We first analyse the case when $N-1 \in \{i_1, \dots, i_{N-2}\}$ and prove the first set of equations in the statement. Expanding the equations \eqref{e:tilderho_neigh} in all possible choices of indexes, one finds a overdetermined system of equations, and expressing it in its matricial form, it is always possible to find a square submatrix whose determinant is equal to $\det\widetilde{g}^{\partial}\neq 0$. This implies that all the variables must be zero.\\
    Let now $N-1 \notin \{i_1, \dots, i_{N-2}\}$. If $N, N-1 \notin \{j_1, \dots, j_{N-2}\}$ no equations are generated. Let then $N \in \{j_1, \dots, j_{N-2}\}$ or $N-1 \in \{j_1, \dots, j_{N-2}\}$ but not $N, N-1 \in \{j_1, \dots, j_{N-2}\}$. We proceed as in the previous case and obtain a system of equations whose only solution is the zero one. Hence we deduce the second and the third set of equations in the statement. Let now $N, N-1 \in \{j_1, \dots, j_{N-2}\}$. Expanding equations \eqref{e:tilderho_neigh} we get
    \begin{align*}
        [\tau, \tilde{e}]^{N N-1 \mu_3 \dots \mu_{N-2}}_{\mu_1 \mu_2 \mu_3 \dots \mu_{N-2}}= \sum_{\sigma} & \tau_{\mu_{\sigma(1)}\dots \mu_{\sigma(N-3)}}^{N N-1 \mu_3 \dots \mu_{N-2}\mu_1} \widetilde{g}^{\partial}_{\mu_1 \mu_{\sigma(N-2)}} \\
        + & \tau_{\mu_{\sigma(1)}\dots \mu_{\sigma(N-3)}}^{N N-1 \mu_3 \dots \mu_{N-2}\mu_2} \widetilde{g}^{\partial}_{\mu_2 \mu_{\sigma(N-2)}}=0.
    \end{align*}
     Inverting some of the equation exploiting the properties of $\widetilde{g}^{\partial}$ we can express the components $\tau_{\mu_3\dots \mu_{(N-3)}\mu_2}^{N N-1 \mu_3 \dots \mu_{N-2}\mu_1}$ with $\mu_1 < \mu_2$ in function of the components of $\widetilde{g}^{\partial}$, 
    $\tau_{\mu_3\dots \mu_{(N-3)}\mu_1}^{N N-1 \mu_3 \dots \mu_{N-2}\mu_2}$ (with $\mu_1 < \mu_2$) and $\tau_{\mu_2\dots \mu_{(N-3)}}^{N N-1 \mu_2 \dots \mu_{N-3}}$.
    
\end{proof}

\begin{proof}[Proof of Lemma \ref{lem:relationSandT}]
    From the proof of Lemma \hyperref[lem:kernel12-21]{\ref*{lem:We_boundary}.(\ref*{lem:kernel12-21}}), the free components of an element in $\mathcal{T}$ are:
        \begin{align*}
            & X^{i_1}_{N-1 i_2} \quad  1 \leq i_1, i_2 \leq N-2, \;  i_1 \neq i_2\\
            & X^{i}_{i N-1} \quad   1 \leq i \leq N-2 
        \end{align*} 
    such that 
        \begin{align*}
             X^{i_1}_{j i_2} = - X^{i_2}_{j i_1}; \quad 
             \sum_{\mu=1}^{N-1} X^{\mu}_{\mu j} =0.
        \end{align*}
    From Proposition \ref{prop:components_of_tau} the free components of an element $\tau \in \mathcal{S}$ are $Y_{\mu}$ and $X_{\mu_1}^{\mu_2}$ satisfying
\begin{align*}
    \sum_{\mu=1}^{N-2} Y_{\mu} =0 \text{ and } 
    X_{\mu_1}^{\mu_2} =- X_{\mu_2}^{\mu_1}
\end{align*}
for $\mu_1, \mu_2= 1 \dots N-2$.
Let us now consider some particular choices of $\tau$. First we consider $\tau$ such that the only nonzero components are $\tau_{\mu_1}^{\mu_2}=-\tau_{\mu_2}^{\mu_1}$ for some particular $\mu_1$ and $\mu_2$. Then
\begin{align*}
    \int_{\Sigma} \tau \alpha & =  \int_{\Sigma}( X_{\mu_1}^{\mu_2} \alpha^{\mu_1}_{N-1 \mu_2}
    + X_{\mu_2}^{\mu_1} \alpha^{\mu_2}_{N-1 \mu_1})\mathbf{V} =  \int_{\Sigma} X_{\mu_1}^{\mu_2} (\alpha^{\mu_1}_{N-1 \mu_2} - \alpha^{\mu_2}_{N-1 \mu_1})\mathbf{V} =0.
\end{align*}
Hence we deduce that $\alpha^{\mu_1}_{N-1 \mu_2} - \alpha^{\mu_2}_{N-1 \mu_1}=0$. Furthermore the components of $p_{\mathcal{T}}(\alpha)$ satisfy $p_{\mathcal{T}}\alpha^{\mu_1}_{N-1 \mu_2} + p_{\mathcal{T}}\alpha^{\mu_2}_{N-1 \mu_1}=0$. Hence we conclude $p_{\mathcal{T}}\alpha^{\mu_1}_{N-1 \mu_2}=0$ for all $\mu_1$ and $\mu_2$. 
Now consider $\tau$ such that the only nonzero components are $Y_\mu$. 
Hence now 
\begin{align*}
    \int_{\Sigma} \tau \alpha & =  \int_{\Sigma} \sum_{\mu=1}^{N-2} Y_{\mu} \alpha^{\mu}_{N-1 \mu}\mathbf{V}
     =  \int_{\Sigma} \sum_{\mu=1}^{N-3} Y_{\mu} (\alpha^{\mu}_{N-1 \mu} - \alpha^{N-2}_{N-1 N-2})\mathbf{V}=0.
\end{align*}
By the arbitrariness of $\tau$ we deduce that  $\alpha^{\mu}_{N-1 \mu} - \alpha^{N-2}_{N-1 N-2}=0$ for each $\mu=1,\dots N-3$. Furthermore the components of $p_{\mathcal{T}}(\alpha)$ satisfy $\sum_{\mu=1}^{N-2} p_{\mathcal{T}}(\alpha)^{\mu}_{N-1 \mu}=0$. Hence we deduce that $ p_{\mathcal{T}}(\alpha)^{\mu}_{N-1 \mu}=0$ for all $\mu$. This proves the claim.
\end{proof}

\begin{proof}[Proof of Lemma \ref{lem:[tau,e]inImW}]
Let $\tau \in \mathcal{S}$. Then we want to prove that 
$[\tau, e] \in \Ima W_{N-3}^{\partial, (1,1)}$.
Using the results of Lemma \ref{lem:We_boundary}, we know that the free components of $\Ima W_{N-3}^{\partial, (1,1)}$ are 
\begin{align*}
    X_{\mu_1 \dots \mu_{N-2}}^{\mu_1 \dots \mu_{N-2}}, \; X_{\mu_1 \dots \mu_{N-3}\mu_{N-2}}^{\mu_1 \dots \mu_{N-3}\mu_{N-1}} \; \text{and} \; X_{\mu_1 \dots \mu_{N-3}N}^{\mu_1 \dots \mu_{N-3}\mu_{N-2}}
\end{align*}
such that $X_{\mu_1 \dots \mu_{N-3}N}^{\mu_1 \dots \mu_{N-3}\mu_{N-2}}=X_{\mu'_1 \dots \mu'_{N-3}N}^{\mu'_1 \dots \mu'_{N-3}\mu_{N-2}}$.

From Proposition \ref{prop:components_of_tau} we deduce that the free components of $\tau \in \mathcal{S}$ are 
\begin{align*}
    \tau^{NN-1 \mu_1 \dots \mu_{N-3}}_{\mu_1 \dots \mu_{N-3}} \; \text{and} \; \tau^{N N-1 \mu_1 \dots \mu_{N-4} \mu_{N-2}}_{\mu_1 \dots \mu_{N-4}\mu_{N-3}}
\end{align*}
such that
\begin{align*}
    \sum_{\mu_i=1}^{N-2}\tau^{NN-1 \mu_1 \dots \mu_{N-3}}_{\mu_1 \dots \mu_{N-3}}=0, \\
    \tau^{N N-1 \mu_1 \dots \mu_{N-4} \mu_{N-2}}_{\mu_1 \dots \mu_{N-4}\mu_{N-3}}+ \tau^{N N-1 \mu_1 \dots \mu_{N-4} \mu_{N-3}}_{\mu_1 \dots \mu_{N-4}\mu_{N-2}}=0.
\end{align*}
Recalling that $[\tau, \tilde{e}]=0$, we deduce that $[\tau, e]$ has components \footnote{We use here the same trick of footnote \ref{foot:brackets} but since $\tau$ can have components in the direction $N-1$ in the standard basis, the metric is the one of the bulk and not the one of the boundary, In particular, since we diagonalized the metric on the boundary we can choose coordinates on the bulk such that $g$ has the form
\begin{align*}
    g = \left( 
    \begin{array}{ccccc}
         \pm 1 & \hdots & 0 & 0 & 0 \\
         \vdots & \ddots & \vdots & \vdots & \vdots \\
         0 & \hdots & \pm 1 & 0 & 0 \\
         0 & \hdots & 0 & 0 & 1 \\
         0 & \hdots & 0 & 1 & 0 \\
    \end{array}
    \right).
\end{align*}}
\begin{align*}
    [\tau, e]_{\mu_1 \dots \mu_{N-3} N-1}^{\nu_1 \dots \nu_{N-2}}= \tau_{\mu_1 \dots \mu_{N-3}}^{\nu_1 \dots \nu_{N-2} N}.
\end{align*}
Plugging into this expression the free components of $\tau$ we get the free components of $[\tau, e]$:
\begin{align*}
    [\tau, e]^{N-1 \mu_1 \dots \mu_{N-3}}_{N-1 \mu_1 \dots \mu_{N-3}} \; \text{and} \; [\tau, e]^{N-1 \mu_1 \dots \mu_{N-4} \mu_{N-2}}_{N-1 \mu_1 \dots \mu_{N-4}\mu_{N-3}}
\end{align*}
such that
\begin{align*}
    \sum_{\mu_i=1}^{N-2}[\tau, e]^{N-1 \mu_1 \dots \mu_{N-3}}_{N-1 \mu_1 \dots \mu_{N-3}}=0, \\
    [\tau, e]^{N-1 \mu_1 \dots \mu_{N-4} \mu_{N-2}}_{N-1 \mu_1 \dots \mu_{N-4}\mu_{N-3}}+ 
    [\tau, e]^{N-1 \mu_1 \dots \mu_{N-4} \mu_{N-3}}_{N-1 \mu_1 \dots \mu_{N-4}\mu_{N-2}}=0.
\end{align*}
It is straightforward to  check that these components are in the image of $ W_{N-3}^{\partial, (1,1)}$.

\end{proof}

\begin{proof}[Proof of Corollary \ref{cor:components_of_W-1[tau,e]}]
Using the standard basis we have that 
\begin{align*}
    (X \wedge e^{N-3})_{i \mu_1 \dots \mu_{N-3}}^{j \mu_1 \dots \mu_{N-3}} &= X_{i}^{j} \text{ for } i \neq j \\
    (X \wedge e^{N-3})_{ \mu_1 \dots \mu_{N-2}}^{\mu_1 \dots \mu_{N-2}} &= \sum_{\mu} X_{\mu}^{\mu}  \text{ with } \mu \in \{\mu_1 \dots \mu_{N-2}\}.
\end{align*}
Comparing these expressions with the ones in the proof Lemma \ref{lem:[tau,e]inImW} we deduce that $[W_{N-3}^{-1}([\tau,e])]_{\mu_1}^{\mu_2}  \propto X_{\mu_1}^{\mu_2}$ and that
\begin{align*}
    \sum_{\mu=1, \mu \neq \nu}^{N-1} [W_{N-3}^{-1}([\tau,e])]_{\mu}^{\mu} = Y_{\nu}.
\end{align*}
Summing for $\nu=1 \dots N-1$ and remembering that $Y_{N-1}=0$ and that $\sum_\nu Y_\nu =0$ we deduce the claim.
\end{proof}

\begin{lemma}\label{lem:vanishofBD-1BT}
    Let $D$ be an invertible matrix such that the inverse does not contain derivatives and let $B$ some matrix proportional to an odd parameter $\lambda$ and not containing derivatives. Then $BD^{-1}B^T=0$.
\end{lemma}
\begin{proof}
The key point of the proof is that every term containing $\lambda^2$ vanishes since $\lambda$ is an odd quantity. Now, by hypothesis every term in $B D^{-1} B^T$ does not contain derivatives and since this expression is quadratic in $\lambda$ vanishes because of the previous consideration.
\end{proof}

\printbibliography
\end{document}